\documentclass[11pt]{article}
\usepackage[margin=1in]{geometry}
\usepackage{amsmath,amssymb,amsthm,mathtools}
\usepackage{xcolor}
\usepackage{tikz}
\usetikzlibrary{arrows.meta,positioning,calc}
\definecolor{diagramblue}{RGB}{30,83,126}
\definecolor{diagramorange}{RGB}{159,77,19}
\tikzset{every picture/.style={line width=.65pt},
  diagram box/.style={draw,rounded corners=2pt,align=center,inner sep=5pt},
  diagram arrow/.style={-{Stealth[length=2mm]}}}
\usepackage{hyperref}
\usepackage{authblk}
\hypersetup{colorlinks=true,linkcolor=blue!45!black,citecolor=blue!45!black,urlcolor=blue!45!black}
\usepackage[backend=biber,style=numeric]{biblatex}
\newtheorem{theorem}{Theorem}[section]
\newtheorem{lemma}[theorem]{Lemma}
\newtheorem{proposition}[theorem]{Proposition}
\newtheorem{corollary}[theorem]{Corollary}
\newtheorem{definition}[theorem]{Definition}
\newtheorem{remark}[theorem]{Remark}

\providecommand{\ket}[1]{| #1 \rangle}
\providecommand{\bra}[1]{\langle #1 |}
\providecommand{\braket}[2]{\langle #1 | #2 \rangle}
\providecommand{\N}{\mathcal{N}}
\providecommand{\QMA}{\mathrm{QMA}}
\providecommand{\QMAone}{\mathrm{QMA}_1}
\providecommand{\Fcal}{\mathcal{F}}
\providecommand{\Id}{\mathbb{I}}

\title{\bf Zero-Energy Problems for Supersymmetric Hamiltonians on a Chain Are $\QMAone$-Complete}

\author{Yibin Wang}
\affil[1]{Graduate School of Mathematics, Nagoya University, Nagoya, 464-8601, Aichi, Japan.\\
\texttt{yibinw0210@gmail.com}}

\date{}

\begin{document}
\maketitle

\begin{abstract}\noindent
We determine the computational complexity of deciding whether a supersymmetric quantum system
on a one-dimensional chain has a zero-energy ground state.  We prove $\QMAone$-completeness
for two input models, formulated in terms of Hermitian and nilpotent supercharges, respectively.
In the Hermitian model, both the supercharge terms and the Hamiltonian have bounded range on
the chain; in the nilpotent model, this geometric restriction applies only to the Hamiltonian.
The hardness reduction from quantum satisfiability on a chain preserves exact zero energy and
gives an inverse-polynomial ground-energy lower bound for negative instances.  A common
verification method establishes membership in $\QMAone$ while allowing cancellations between
supercharge terms.
\end{abstract}

\section{Introduction}
\label{sec:intro}

Supersymmetric quantum mechanics gives zero energy a special algebraic meaning.  In the minimal
formulation, a Hermitian supercharge $Q$ defines $H=Q^2$, so a supersymmetric ground state is
exactly a nonzero vector annihilated by $Q$.  In the extended formulation, a nilpotent
supercharge $\mathcal Q$ defines $H=\{\mathcal Q,\mathcal Q^\dagger\}$, and a zero mode is
annihilated by both $\mathcal Q$ and $\mathcal Q^\dagger$.  The connection between
supersymmetric ground states and cohomology goes back to Witten's use of supersymmetric quantum
mechanics in Morse theory~\cite{Witten1982}.

From a computational viewpoint, exact zero energy differs from ordinary ground-energy
estimation.  The Local Hamiltonian problem distinguishes two separated energy
ranges~\cite{Kitaev2002}.  Quantum satisfiability asks whether all local projector constraints
can be met exactly~\cite{Bravyi2011}; already its three-local version is
$\QMAone$-complete~\cite{Gosset2013}.  Frustration-free structure can also permit quadratic spectral-gap
amplification~\cite{SommaBoixo2013}.

Geometric locality imposes a second restriction.  The usual $k$-locality condition limits how
many degrees of freedom a term may touch, whereas geometric locality also requires those degrees
of freedom to be nearby.  The Local Hamiltonian problem remains $\QMA$-complete on a
two-dimensional square lattice~\cite{Oliveira2008} and on a one-dimensional nearest-neighbour
chain when no constant gap is promised~\cite{Aharonov2011}.  For constant-gap chains with a unique ground state, the area
law gives matrix-product-state approximations~\cite{Hastings2007}, and the ground state can be
approximated in randomized polynomial time~\cite{LandauVaziraniVidick2015}.  One dimension therefore
contains both tractable and hard regimes.

Exact quantum satisfiability is the closest one-dimensional precedent.  Nagaj proved
$\QMAone$-completeness on a nearest-neighbour chain~\cite{Nagaj2008}.  Later work found
additional $\QMAone$-complete quantum-satisfiability families, including systems with small local
state spaces and systems arranged on a line~\cite{Rudolph2025}.  A supersymmetric reduction must also preserve exact zero energy and the
one-dimensional interaction geometry.

Cade and Crichigno proved that supersymmetric local-Hamiltonian problems are hard when locality
limits term size without imposing a lattice geometry~\cite{Cade2024}.  They asked whether this
hardness survives geometric locality.  We treat the exact-zero form of that question on a chain,
for both Hermitian and explicitly encoded nilpotent supercharges.

\paragraph{Contributions.}

The following statements compare Cade and Crichigno's $\N=2$ classification without a geometric
restriction with the exact-zero, geometrically local classification proved here.

\begin{theorem}[Cade and Crichigno]
\label{thm:cc-hamiltonian}
For some constant $k$, the general-threshold supersymmetric $k$-local Hamiltonian problem is
$\QMA$-complete~\cite[Thm.~1]{Cade2024}.
\end{theorem}

Their energy promise allows a positive YES threshold.  Here the YES energy is exactly zero,
separated from NO energies by an inverse-polynomial lower bound, and the hard Hamiltonians have
bounded interaction range on a chain.  In both input models, the local supercharge matrices
are listed explicitly with exact coefficients and a specified NO energy lower bound.

\begin{theorem}[this paper]
\label{thm:cmp-zero-energy}
For some constant $k$, the zero-energy problem for $k$-local supersymmetric Hamiltonians is
$\QMAone$-complete in both the Hermitian and nilpotent formulations, even when the Hamiltonian
is geometrically local on a one-dimensional chain.
\end{theorem}

The exact input models are specified in Definitions~\ref{def:1glsp}
and~\ref{def:explicit-table-exact-zero}.  In the Hermitian formulation, the supercharge is
parity-odd and its listed terms have bounded chain support.  In the nilpotent formulation,
the chain-local restriction of Definition~\ref{def:target-chain-restriction} applies only to
$\Delta=\{\mathcal Q,\mathcal Q^\dagger\}$; the supercharge and its individual terms retain
the general geometry allowed by the input model.

The hardness construction attaches distinct Majorana modes to the local positive terms of a
frustration-free chain.  Anticommutation removes cross terms between disjoint supports when the
Hermitian supercharge is squared.  A spectral estimate then preserves the zero-energy condition
and bounds the ground energy of negative instances.  This estimate has a quadratic energy loss,
which is sharp for the single-Majorana construction.  A parity-preserving two-mode extension
produces the nilpotent supercharge with the same ground energy and Hamiltonian interaction range.

Perfect-completeness verification requires control of the global zero-mode equations.  Local
supercharge terms can cancel on a state, so testing their individual kernels can reject a valid
zero mode.  We construct exact sparse linear constraints that retain these cancellations and
prove the coefficient and conditioning bounds needed for Rudolph's exact-nullspace
verifier~\cite{Rudolph2024gateset}.  The resulting common-kernel theorem supplies one containment
argument for both formulations.  Together with the chain constructions, it gives the
classifications in Theorems~\ref{thm:n1-tiers} and~\ref{thm:explicit-n2-completeness}, where the
input encodings and quantitative bounds are stated precisely.

Section~\ref{sec:prelim} introduces the complexity and locality conventions.
Sections~\ref{sec:ferm}, \ref{sec:containment}, and~\ref{sec:n1} establish the Hermitian result,
and Section~\ref{sec:explicit-n2} establishes the nilpotent result.  Supporting constructions and
quantitative bounds are collected in the appendices.
\section{Preliminaries}
\label{sec:prelim}

\subsection{The class \texorpdfstring{$\QMAone$}{QMA1} and frustration-free Hamiltonians}
\label{sec:qma1}

A promise problem lies in $\QMA$ if a polynomial-time quantum verifier accepts a
correct witness with probability at least $c$ and any witness for a negative
instance with probability at most $s$, with $c-s\ge 1/\mathrm{poly}$. The
restriction $\QMAone$ demands perfect completeness, meaning $c=1$, and thus couples the
class to exact ground-state energies.

\paragraph{Verification convention.}
Throughout, $\QMAone$ denotes exact verification over the fixed universal gate set
\[
\mathcal G=\{\widehat H,T,\mathrm{CNOT}\},\qquad
\widehat H=2^{-1/2}\begin{pmatrix}1&1\\1&-1\end{pmatrix},
\qquad
T=\operatorname{diag}(1,e^{i\pi/4})
=\operatorname{diag}\!\left(1,\frac{1+i}{\sqrt2}\right).
\]
Cade and Crichigno use this convention~\cite{Cade2024}.
Gosset and Nagaj use the same gate set for exact projector
measurements~\cite{Gosset2013}.
The verifier families are polynomial-time uniform and have size polynomial in the input
encoding length $S$. Every circuit equality used for this class is an exact matrix
identity, including global phase.

A Hamiltonian $H=\sum_jH_j$ with positive semidefinite terms is frustration-free when one
ground state annihilates every $H_j$.  For local constraints $h_a$ with input-supplied exact
rejection circuits over $\mathcal G$, their rejection probabilities are exact quadratic forms.
Sequential repetition, or Marriott--Watrous amplification~\cite{MarriottWatrous2005}, improves
an inverse-polynomial soundness gap while preserving perfect completeness.  A work qubit is
\emph{clean} when it is initialized in $\ket0$ and returned exactly to $\ket0$.
The NO promise $E_0(H)\ge\delta$ bounds energy above zero; it does not bound the spacing
between the ground energy and the next distinct eigenvalue of a NO instance.
The path-Laplacian gap used in the history-state analysis is a separate auxiliary estimate.

\begin{lemma}[Exact constant-soundness normalization]
\label{lem:ambient-normalization}
For each fixed language in $\QMAone$, there is a polynomial-time reduction to verifier circuits
over $\mathcal G$ with perfect completeness and NO acceptance at most $1/2$.
\end{lemma}

One frustration-free Hamiltonian recurs throughout: the Feynman-Kitaev history Hamiltonian of a
circuit~\cite{Kitaev2002}.  Its propagation and boundary terms form
$H_{in}=H_{\mathrm{prop}}+H_{\mathrm{pen}}$, which has zero energy exactly when some witness is
accepted with certainty.  Kitaev's unary clock records time $t$ by the wall in
$\mathtt1^t\mathtt0^{L-t}$ and penalizes illegal $\mathtt0\mathtt1$ pairs locally.  Gosset and
Nagaj give an exact three-local quantum-SAT construction~\cite{Gosset2013}.  The reductions below
use Nagaj's uniform $d=11$ nearest-neighbor line source stated in Theorem~\ref{thm:rgn}.

\subsection{Fermions and Majorana operators}

A system of $N$ fermionic modes is characterized by creation operators $c_{i}^{\dagger}$ and annihilation operators $c_{i}$, where $i\in\{1,\dots,N\}$, obeying
the canonical anticommutation relations (CAR), and its Fock space carries the integer
grading by total occupation number. A
complex mode $c$ splits into two Hermitian Majoranas $\Gamma_A=c+c^\dagger$ and
$\Gamma_B=i(c^\dagger-c)$, and a set $\{\Gamma_a\}$ obeys
\begin{equation}
\{\Gamma_a,\Gamma_b\}=2\delta_{ab}\Id,\qquad \Gamma_a^\dagger=\Gamma_a,
\qquad \Gamma_a^2=\Id .
\end{equation}
Distinct Majoranas anticommute, and the algebra they generate is a Clifford
algebra. The occupation operator of mode $i$ is $n_i=c_i^\dagger c_i$, and the total fermion
number operator is $\hat N=\sum_{i=1}^{N}n_i$. The parity operator $P=(-1)^{\hat N}$ separates operators into even ones,
which commute with $P$, and odd ones, which anticommute. Physical observables are
even. Two even operators on disjoint
supports commute, two odd operators on disjoint supports anticommute, and an odd
operator built as the product of an even physical operator with a single Majorana
inherits the second behavior.

\subsection{Minimal and extended supersymmetry}
\label{sec:susy-types}

\paragraph{Minimal supersymmetry.}
A fermionic system has $\N=1$ supersymmetry when an odd Hermitian supercharge $Q$ gives
$H=Q^2$.  The identity $PQP=-Q$ makes the Hamiltonian parity even.  It is positive
semidefinite, and its zero modes are exactly $\ker Q$.

\paragraph{The Hermitian-square problem.}
The exact-zero question also makes sense for a general Hermitian $Q$ that is not parity
homogeneous.  We call this the ungraded Hermitian-square problem; intrinsic minimal
supersymmetry is its parity-odd restriction.  The reduction in
Theorem~\ref{thm:minimal-hardness} produces parity-odd inputs.  For $Q=\sum_jQ_j$, the Hamiltonian
$H=\sum_jQ_j^2+\sum_{j<k}\{Q_j,Q_k\}$ contains cross terms that may be indefinite and
noncommuting, so it need not have a termwise frustration-free decomposition.  The reduction
therefore uses the global spectral map of Lemma~\ref{lem:spectral}.

\paragraph{Extended supersymmetry.}
A fermionic system has $\N=2$ supersymmetry when an odd supercharge $Q$ satisfies $Q^2=0$
and gives $H=\{Q,Q^\dagger\}$.  Write $H(Q)=\ker Q/\operatorname{im}Q$; when the Hilbert space is graded
and $Q$ raises degree, this quotient is taken degree by degree.  The Laplacian identity
\[
 \langle\psi|H|\psi\rangle=\lVert Q\psi\rVert^2+\lVert Q^\dagger\psi\rVert^2
\]
shows that its zero modes are closed ($Q\psi=0$) and co-closed ($Q^\dagger\psi=0$).
Nilpotency gives $\operatorname{im}Q\subseteq\ker Q$: cohomology identifies closed vectors
that differ by an exact vector $Q\chi$.  In finite dimensions each class has a unique
representative in $\ker Q\cap\ker Q^\dagger$, orthogonal to $\operatorname{im}Q$.
These harmonic representatives are the zero-energy states tested below.

\subsection{Geometric locality}

We place modes on a lattice $\Lambda\subset\mathbb{Z}^D$ with the graph metric and take
range bounds independent of system size.  At bounded mode density, geometric locality gives
constant mode arity; a bound on arity alone imposes no spatial constraint.

\begin{definition}[Mode allocation and geometric locality]
\label{def:mode-allocation-locality}
Let $\mathfrak M$ be the finite enumerated mode set and let $\Lambda$ be the explicitly listed
lattice sites.  A mode allocation is a disjoint union
\[
\mathfrak M=\bigsqcup_{s\in\Lambda}\mathfrak M_s,
\]
where every fibre $\mathfrak M_s$ is ordered and may be empty.  It determines the total map
$\pi:\mathfrak M\to\Lambda$ by $\pi(\mu)=s$ exactly when $\mu\in\mathfrak M_s$; concatenating
the fibres in site order gives the site-major mode order.  The \emph{mode arity} of a term $T$ is
the number of fermionic modes in its mode support, and its \emph{site support} is
$\pi(\operatorname{supp}_{\rm mode}T)$.  The \emph{metric
diameter} is the largest lattice distance between two support sites.  On a chain, the
\emph{span}, or window width, is the number of sites in the smallest containing interval, hence
one plus the diameter.  Thus a range bound $R$, meaning diameter less than $R$, gives chain span
at most $R$.
\end{definition}

For a fermionic operator, \emph{CAR-locality} measures the site support of the operator in the
fermionic algebra before a Jordan--Wigner transformation. Spin/Pauli locality measures support in
the chosen spin representation. On a chain, the \emph{realized span} of a constructed family is
the maximum span among its nonzero terms.

\begin{definition}[Locality data for anchored families]
\label{def:anchored-locality-data}
For a listed family with containing intervals $I_c$ and designated anchors $a(c)$, write
\[
w=\max_c|I_c|,\qquad
D_{\mathrm{site}}=\max_s\#\{c:s\in I_c\},\qquad
A=\max_s\#\{c:a(c)=s\},
\]
and define the structural overlap degree by
\[
\Delta_{\mathrm{ov}}=\max_c\#\{d\ne c:I_c\cap I_d\ne\varnothing\}.
\]
The range, site load, anchor multiplicity, and overlap degree are recorded separately for each
constructed family.
\end{definition}

\begin{lemma}[Overlap-counting bound]
\label{lem:overlap-counting}
For every listed anchored family,
\[
\Delta_{\mathrm{ov}}\le wD_{\mathrm{site}}-1.
\]
\end{lemma}
\begin{proof}
For each $c$,
\[
\#\{d\ne c:I_c\cap I_d\ne\varnothing\}
\le \sum_{s\in I_c}\#\{d:s\in I_d\}-1
\le |I_c|D_{\mathrm{site}}-1
\le wD_{\mathrm{site}}-1.
\]
Taking the maximum over $c$ proves the claim.
\end{proof}

\begin{lemma}[Square roots preserve locality and parity]
\label{lem:sqrt}
Let $h$ be a positive semidefinite operator on a finite-dimensional space. Its
unique positive semidefinite square root $\sqrt h$ is supported within the
modes that support $h$, so $\sqrt h$ is geometrically local with the same range.
If in addition $h$ is even under fermion parity, then $\sqrt h$ is even.
\end{lemma}
\begin{proof}
The spectrum of $h$ is finite, so there is a polynomial $p$ with $p(\lambda)=\sqrt\lambda$
at every eigenvalue $\lambda$ of $h$, and then $\sqrt h=p(h)$. The square root therefore lies
in the unital algebra generated by $h$ itself, hence in the algebra of the modes that support
$h$, and the locality is strict, with no decaying tail. The same expression gives the second
claim: $[h,P]=0$ makes $P$ commute with every power of $h$, so
$[\sqrt h,P]=[p(h),P]=0$.  This algebraic formulation is needed for fermionic support: the
supporting modes generate a local CAR subalgebra, whereas a diagonalizing unitary need not
belong to that subalgebra.
\end{proof}

\subsection{A one-dimensional history-state reduction}
\label{sec:input}

\begin{theorem}[Nagaj's one-dimensional history reduction]
\label{thm:rgn}
Nagaj proves that quantum SAT on a line of sites of local dimension $d=11$ is
$\QMAone$-complete.  The reduction is
uniform and polynomial time, its constraints are nearest-neighbor projectors, YES instances have
an accepting history in their common kernel, and NO instances have inverse-polynomial energy
\cite[Secs.~4.3.1--4.3.2]{Nagaj2008}.
\end{theorem}

We use this architecture and NO bound through the normalized source family below, whose schedule,
local projector grouping, and exact circuit realization are specified explicitly.
The projector counts and raw energy bound use the five-qubit even-weight encoding of
Section~\ref{sec:ferm}, including its penalties for unused codewords.

We first keep every local projector occurrence separately; this is the unmerged, or raw,
presentation.  The bondwise merge groups occurrences assigned to the same boundary or physical
bond and replaces each group by the projector onto the span of their ranges.
Let $n$ be the number of verifier wires and $K$ the number of nearest-neighbour gate rounds
after routing and normalization.  Appendix~\ref{app:source-normalization-count} records the
source size, raw projector count, and a computable polynomial $G_{\mathrm{gap}}(n,K)$.
The raw Hamiltonian has the accepting history space as its YES kernel and obeys
$H_{\mathrm{rec}}\succeq G_{\mathrm{gap}}^{-1}\Id$ on NO instances
(Proposition~\ref{prop:ambient-raw-gap}).  The following merge keeps its kernel and loses
at most a factor of eight in this lower bound.

\begin{lemma}[Kernel and gap under bondwise merging]
\label{lem:source-grouping}
Let $H_{\mathrm{rec}}$ be the raw projector sum of
Lemma~\ref{lem:source-normalization}, and let $\epsilon_{\mathrm{src}}$ denote its NO-instance
lower bound.  For the family first normalized by Lemma~\ref{lem:ambient-normalization},
Proposition~\ref{prop:ambient-raw-gap} permits the explicit choice
$\epsilon_{\mathrm{src}}=G_{\mathrm{gap}}^{-1}$.  Partition the raw projector occurrences by
their assigned boundary or physical bond.  For each group $b$, let
$\widehat\Pi_b$ project onto the span of the ranges of its $m_b$ raw projectors, and define the
bondwise-merged operator
\[
H_{\mathrm{bw}}:=\sum_b\widehat\Pi_b.
\]
The largest raw bond multiplicity is $m_{\mathrm{bond}}=\max_b m_b=8$.  The merge produces
$M=2nK+2$ projectors and preserves the kernel.  Its NO ground-energy lower bound
$\epsilon_{\mathrm{ff}}$ satisfies
\[
\epsilon_{\mathrm{ff}}\ge\epsilon_{\mathrm{src}}/8.
\]
\end{lemma}
\begin{proof}
For a group carrying projectors $\Pi_{b,1},\ldots,\Pi_{b,m_b}$, let $\widehat\Pi_b$ project onto
the span of their ranges.  Their common kernel is unchanged and
$\sum_\ell\Pi_{b,\ell}\preceq m_b\widehat\Pi_b\preceq8\widehat\Pi_b$.  Summing over groups
gives $H_{\mathrm{rec}}\preceq8H_{\mathrm{bw}}$.  The count in
Appendix~\ref{app:source-normalization-count} has $N+1=2nK+2$ groups, which proves the merged
count and the factor-eight gap transfer.
\end{proof}

\begin{lemma}[Exact nearest-neighbor routing]
\label{lem:source-routing}
Every problem in $\QMAone$ reduces in polynomial time to the fixed source family of
Theorem~\ref{thm:rgn}, with the normalization of Lemma~\ref{lem:source-normalization}.
Each scheduled two-site gate has both an exact gate sequence over
$\mathcal G$ and its entrywise-equal exact $4\times4$ matrix.
\end{lemma}
\begin{proof}
Start with the normalized verifier of Lemma~\ref{lem:ambient-normalization} and apply the size
adjustments of Lemma~\ref{lem:source-normalization} to obtain $n,K\ge2$.  Embed a one-qubit
gate in a two-site slot as $H_L$, $H_R$, $T_L$, or $T_R$, using the notation of
Appendix~\ref{app:explicit-n2-reduction}.
Replace each non-neighbour two-qubit gate using the exact identity
\[
\mathrm{SWAP}_{j,j+1}=
\mathrm{CNOT}_{j,j+1}\,
\mathrm{CNOT}_{j+1,j}\,
\mathrm{CNOT}_{j,j+1},
\]
and relocate the output qubit to the final wire. In left-to-right sweeps, position $(r,j)$
carries the required adjacent operation or the identity, together with its exact matrix in
the ordered basis $00,01,10,11$ and an entrywise-equal exact gate sequence over $\mathcal G$.
These replacements and the final placement of the output check preserve acceptance exactly with
polynomial overhead.
\end{proof}

\begin{lemma}[Exact reflections for the history-state projectors]
\label{lem:fieldtransfer}
For every unmerged or bondwise-merged projector $\Pi$ in the history-state construction,
the reflection $I-2\Pi$ and its controlled version $\operatorname{diag}(I,I-2\Pi)$ have
phase-sensitive exact circuits over $\mathcal G$ using only eight clean work qubits.
\end{lemma}
\begin{proof}
Appendix~\ref{app:source-reflections} constructs the partial-support blocks and proves the
eight-clean-qubit bound.
\end{proof}

\begin{remark}[History-state families and gap bounds]
The $\Theta(1/L^2)$ estimate used later concerns an auxiliary path-graph
Laplacian~\cite{Kitaev2002}; the raw source bound is given in
Proposition~\ref{prop:ambient-raw-gap}.
\end{remark}

\section{Fermionization on the line}
\label{sec:ferm}

We fermionize the qudit input by a local qubit code followed by open-boundary
Jordan--Wigner. This gives a full-Fock encoding without the local occupation constraints of
Batista and Ortiz~\cite{BatistaOrtiz2001}.

\subsection{Local qudit-to-qubit encoding}

Encode each $d=11$ source site into $m_{11}=\lceil\log_2 11\rceil+1=5$ qubits by a
fixed isometry
\[
\iota:\mathbb{C}^{11}\hookrightarrow(\mathbb{C}^2)^{\otimes5},
\]
sending the eleven source states to distinct even-weight computational words.
Appendix~\ref{app:clock-constructor} fixes the ordered code used in the source construction.
Four qubits have only $2^{4-1}=8$ even-weight words, whereas five have $2^{5-1}=16$;
the parity requirement accounts for the fifth qubit.  Each source site is then unfolded
into five consecutive encoded chain sites, one system mode per site.

For a source projector $h_a$ supported on sites $T_a$, let
$\widetilde h_a=\iota_{T_a}h_a\iota_{T_a}^\dagger$ on the encoded valid subspace and extend it
by zero on the unused local codewords.  Let $\Pi_i^{\mathrm{inv}}$ project onto the unused
subspace of encoded particle $i$.

Assume that the source projector sum $\sum_a h_a$ is frustration-free in a positive instance
and has a declared NO lower bound $0<\epsilon\le1$ in a negative instance.

\begin{lemma}[The local encoding preserves the promise and geometry]
\label{lem:encoding}
The Hamiltonian $\widetilde H=\sum_a\widetilde h_a+\sum_i\Pi^{\mathrm{inv}}_i$ is a sum of
geometrically local
qubit projectors with two properties. A positive instance keeps a common
zero eigenvector, so $\widetilde H$ stays frustration-free with $E_0=0$, and a
negative instance has $E_0\ge\epsilon$.
\end{lemma}
\begin{proof}
On the all-valid subspace, $\sum_a\widetilde h_a$ is isometric to $\sum_a h_a$; hence the
encoded accepting history is a common zero vector, while a negative instance has energy at
least $\epsilon$. Each $\widetilde h_a=\iota h_a\iota^\dagger$ has range and domain in the
valid subspace of its particles and therefore commutes with every invalid-code projector,
so $\widetilde H$ is block diagonal. On its orthogonal complement at
least one $\Pi_i^{\mathrm{inv}}$ contributes one while $\sum_a\widetilde h_a\succeq0$; every
state in this invalid-code block therefore has energy at least $1\ge\epsilon$. A $\widetilde h_a$ occupies one ten-mode
two-site window and each penalty one encoded particle, preserving geometry.
\end{proof}

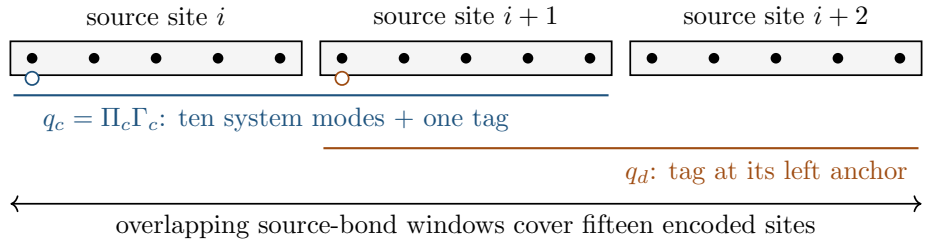
\begin{figure}[htbp]
\centering
\begin{tikzpicture}[x=.82cm,y=.70cm,font=\small]
  \foreach \a/\b/\t in {0/4/{source site $i$},5/9/{source site $i+1$},10/14/{source site $i+2$}} {
    \draw[fill=black!4] (\a-.35,-.32) rectangle (\b+.35,.32);
    \node at ({(\a+\b)/2},.80) {\t};
  }
  \foreach \x in {0,...,14} \fill (\x,0) circle (2pt);
  \draw[diagramblue,thick] (-.3,-.70)--(9.3,-.70);
  \node[anchor=west,diagramblue] (qc) at (0,-1.16) {$q_c=\Pi_c\Gamma_c$: ten system modes + one tag};
  \draw[diagramblue,fill=white] (0,-.38) circle (2.5pt);
  \draw[diagramorange,thick] (4.7,-1.70)--(14.3,-1.70);
  \draw[diagramorange,fill=white] (5,-.38) circle (2.5pt);
  \node[anchor=east,diagramorange] at (14.3,-2.16) {$q_d$: tag at its left anchor};
  \draw[<->] (-.35,-2.75)--(14.35,-2.75);
  \node at (7,-3.17) {overlapping source-bond windows cover fifteen encoded sites};
\end{tikzpicture}
\caption{Mode arity and chain span in the constructed family.  Filled dots are system
modes; open circles are distinct label modes placed at existing encoded sites.  A term
$q_c$ therefore has at most eleven modes but span at most ten.  Two adjacent source-bond
windows cover fifteen encoded sites, so their surviving cross terms fit in this union.
For arbitrary overlapping windows of width ten, the general bound is nineteen sites.
The window picture gives support bounds; the algebra determines which cross terms survive.}
\label{fig:mode-span}
\end{figure}

\subsection{Open-boundary Jordan--Wigner}

Figure~\ref{fig:mode-span} distinguishes the encoded sites from the modes assigned to them.
The $N$ source sites yield $N_q:=5N$ encoded qubits.
Order the $N_q$ encoded qubits along the line and apply
$\sigma_i^+=c_i^\dagger\prod_{j<i}(1-2n_j)$ and $\sigma_i^z=1-2n_i$. On an open
chain this is a unitary $*$-isomorphism of operator algebras, hence isospectral,
and the parity-sector subtleties of a ring do not arise.

\begin{lemma}[Jordan--Wigner preserves the spectrum and locality]
\label{lem:jw}
Let $\tilde H=\sum_r\tilde h_r$ act on contiguous blocks of the open qubit chain,
with each $\tilde h_r$ even under the Jordan--Wigner parity. Its fermionic image
$H_{in}=\sum_r h_r$ has the same spectrum, each $h_r$ is even and keeps the contiguous
support of $\tilde h_r$, and any common zero vector is carried to a common zero vector.
If $\tilde H$ satisfies the frustration-free YES case and NO energy lower bound of
Lemma~\ref{lem:encoding}, then $H_{in}$ satisfies the same promise.
The full occupation-basis Fock space is used: no particle-number
sector is fixed or penalized, and only fermion parity is superselected; both parity sectors
remain in the spectrum.
\end{lemma}
\begin{proof}
A unitary $*$-isomorphism preserves the spectrum and carries a common zero
eigenvector to a common zero eigenvector, which gives isospectrality and
frustration-freeness; it also carries the qubit parity to the fermion parity, so an
even $\tilde h_r$ has an even image. For locality, expand $\tilde h_r$ in the Pauli basis of its
block. Evenness means that every surviving monomial carries an even number of $\sigma^x$ and
$\sigma^y$ letters, say at positions $i_1<\dots<i_{2m}$ inside the block, with $\sigma^z$ letters
elsewhere. Each $\sigma^x$ or $\sigma^y$ at position $i$ contributes the string
$\prod_{j<i}(1-2n_j)$, and the strings pair off from the left,
$\bigl(\prod_{j<i_1}\bigr)\bigl(\prod_{j<i_2}\bigr)=\prod_{i_1\le j<i_2}$ and so on, so the
factors reaching back to site $1$ cancel and what remains,
$\prod_{i_1\le j<i_2}(1-2n_j)\cdots\prod_{i_{2m-1}\le j<i_{2m}}(1-2n_j)$, is supported inside
$[i_1,i_{2m}]$. The $\sigma^z$ letters map to the on-site $1-2n_i$. Every factor of the image
therefore sits inside the block. An odd monomial would keep the uncancelled leading string
$\prod_{j<i_1}(1-2n_j)$ and run to the boundary of the chain, which is exactly what the
even-weight encoding of Lemma~\ref{lem:encoding} rules out.
\end{proof}

\begin{figure}[htbp]
\centering
\begin{tikzpicture}[x=1cm,y=1cm,font=\small]
  \node[anchor=west,font=\small\bfseries] at (0,0) {(a) Even Jordan--Wigner monomial};
  \draw[black!40] (.25,-.85)--(13.8,-.85);
  \fill (4,-.85) circle (2pt) node[above=2pt] {$i_1$};
  \fill (10,-.85) circle (2pt) node[above=2pt] {$i_2$};
  \draw[diagramblue,thick] (.25,-1.5)--(4,-1.5);
  \draw[diagramorange,thick] (.25,-2.2)--(10,-2.2);
  \node[anchor=west] at (10.25,-1.5) {$\prod_{j<i_1} Z_j$};
  \node[anchor=west] at (10.25,-2.2) {$\prod_{j<i_2} Z_j$};
  \node at (1.85,-2.95) {common prefix cancels};
  \draw[diagramblue,very thick] (4,-2.95)--(10,-2.95);
  \node[anchor=west] at (10.25,-2.95) {$\prod_{i_1\le j<i_2} Z_j$};
  \node[anchor=west,font=\small\bfseries] at (0,-3.8) {(b) Disjoint terms in the square};
  \node[diagram box,draw=diagramblue] at (2.7,-4.7) {system window $I_c$\\$q_c=h_c\Gamma_c$};
  \node[diagram box,draw=diagramorange] at (9.7,-4.7) {system window $I_d$\\$q_d=h_d\Gamma_d$};
  \node at (6.9,-5.75) {$I_c\cap I_d=\varnothing\quad\Longrightarrow\quad
    \{q_c,q_d\}=[h_c,h_d]\Gamma_c\Gamma_d=0$};
\end{tikzpicture}
\caption{Two cancellations that preserve locality, with $Z_j=1-2n_j$.
(a) An even monomial has an even number of Jordan--Wigner prefixes.  They pair to leave
strings inside the original block, as illustrated for two odd letters; endpoint factors
also stay in that block.  (b) For the encoded projectors $h_c$, evenness makes them commute
with all distinct label Majoranas.  Disjoint $h_c,h_d$ commute, while the labels
anticommute, eliminating the distant cross term in $Q^2$.  Overlapping terms can survive
but remain in the union of their windows.  These are support statements; the energy bound
uses Lemma~\ref{lem:spectral}.}
\label{fig:local-cancellations}
\end{figure}

\begin{proposition}[Fermionic encoding of the one-dimensional history Hamiltonian]
\label{prop:ferm}
The fixed source of Theorem~\ref{thm:rgn}, with the parameters and bondwise merge of
Lemmas~\ref{lem:source-normalization} and~\ref{lem:source-grouping}, has a polynomial-time
fermionic encoding as an open chain of
\[
M=2nK+2
\]
local projectors $h_c=\Pi_c$, each of which is parity even, with
$H_{in}=\sum_ch_c=H_{\mathrm{bw}}$. Each bondwise-merged projector is supported on at most ten
consecutive encoded modes. For this family, the anchored-family locality data of
Definition~\ref{def:anchored-locality-data} are
$w=10$, $D_{\mathrm{site}}\le3$, $A\le2$, and $\Delta_{\mathrm{ov}}\le4$. The family is frustration-free on YES
instances and on NO instances satisfies
\[
E_0(H_{in})\ge\epsilon_{\mathrm{ff}}\ge\epsilon_{\mathrm{src}}/8.
\]
\end{proposition}
\begin{proof}
Apply the even-weight encoding and Lemma~\ref{lem:jw} to the projectors obtained by the bondwise
merge. The source construction includes its invalid-code completion within the count
$M=2nK+2$. Each encoded two-site window has ten modes and remains contiguous under
Jordan--Wigner. The bondwise-merged source data directly give the stated site load, anchor
multiplicity, and overlap degree. Isospectrality preserves the
frustration-free YES case and the lower bound $\epsilon_{\mathrm{ff}}$.
\end{proof}

\section{Nullspace verification for sparse matrices}
\label{sec:containment}

The hardness construction must be complemented by a verifier for every promised input,
including inputs whose local supercharge terms cancel.  This section builds that verifier
from the global equation $Q\psi=0$ (or the pair $\mathcal Q\psi=\mathcal Q^\dagger\psi=0$).
It encodes coefficient multiplication as sparse relations, retains the sum over local terms,
and converts the resulting relation matrix into the Hermitian input required by the
exact-nullspace verifier.  Figure~\ref{fig:verification-route} summarizes these steps.
Exact sparse access means that reversible circuits enumerate the nonzero entries of a requested
row or column and return their coefficients without approximation.  Throughout,
\[
K_8=\mathbb Q(\zeta_8),\qquad \zeta_8=e^{i\pi/4},
\]
with rational basis $(1,\zeta_8,\zeta_8^2,\zeta_8^3)$, and all verifier circuits use
$\mathcal G$ as defined in Section~\ref{sec:qma1}.

We represent multiplication by an input coefficient using sparse linear equations with
auxiliary variables recording the intermediate values.  We call their collection a history
and the vector of equation errors its residual.  A zero residual means that all equations
hold exactly.  The construction below has a succinct description even when the history has
exponentially many coordinates; a requested equation is accessed by its binary stage label.

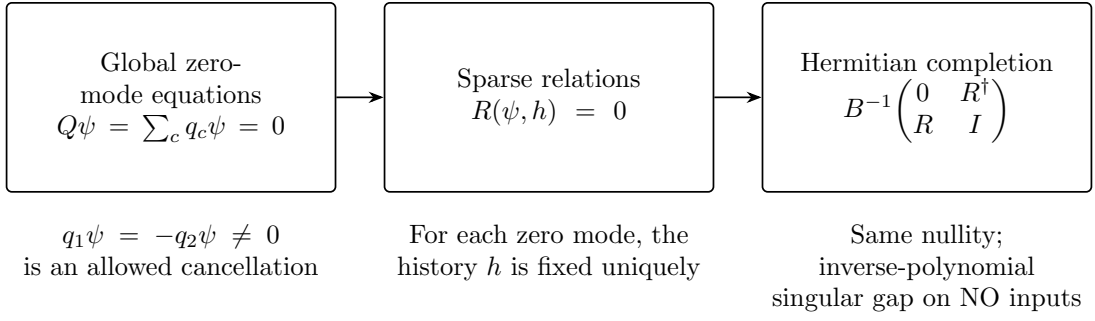
\begin{figure}[htbp]
\centering
\begin{tikzpicture}[x=1cm,y=1cm,font=\small]
  \node[diagram box,text width=4cm,minimum height=2.5cm,anchor=north] (a) at (2,0)
    {Global zero-mode equations\\$Q\psi=\sum_cq_c\psi=0$};
  \node[diagram box,text width=4cm,minimum height=2.5cm,anchor=north] (b) at (7,0)
    {Sparse relations\\$R(\psi,h)=0$};
  \node[diagram box,text width=4cm,minimum height=2.5cm,anchor=north] (c) at (12,0)
    {Hermitian completion\\$B^{-1}\!\begin{pmatrix}0&R^\dagger\\R&I\end{pmatrix}$};
  \draw[diagram arrow] (a)--(b);
  \draw[diagram arrow] (b)--(c);
  \node[align=center,text width=4.25cm,anchor=north] at (2,-2.8)
    {$q_1\psi=-q_2\psi\ne0$\\is an allowed cancellation};
  \node[align=center,text width=4.25cm,anchor=north] at (7,-2.8)
    {For each zero mode, the\\history $h$ is fixed uniquely};
  \node[align=center,text width=4.25cm,anchor=north] at (12,-2.8)
    {Same nullity; inverse-polynomial\\singular gap on NO inputs};
\end{tikzpicture}
\caption{From a global zero mode to exact verification.  The relations include both
coefficient histories and equations summing the local contributions.  Thus cancellation
is retained, while the uniquely determined histories add no independent kernel vectors.
The lower-right identity in the completion gives
$\ker\bigl(\begin{smallmatrix}0&R^\dagger\\R&I\end{smallmatrix}\bigr)
=\ker R\oplus0$.  The stability bounds below make the resulting matrix suitable for
the exact-nullspace verifier.  For a nilpotent input the first box contains both global
equations, with the same subsequent steps.}
\label{fig:verification-route}
\end{figure}

The two input models use this machinery with different coefficient bounds.
The Hermitian model allows a separate binary exponent in each quadratic-dyadic entry,
so a coefficient can be extremely small even when its description is short.
The explicit nilpotent model instead bounds a common denominator and its $K_8$ coordinates
numerically by a polynomial in the input length.  Succinct coefficient histories handle
the former case; exact sparse access handles the latter.  Both retain the full Fock space.

\begin{lemma}[Exact multiplication by a normalized input coefficient]
\label{lem:stable-quadratic-dyadic-multiplier}
Let
\[
\alpha=\frac{a+b\sqrt2+i(c+d\sqrt2)}{2^t},
\qquad a,b,c,d\in\mathbb Z,\quad t\ge0,
\]
be a nonzero entry of an explicitly supplied operator of norm at most one.  Suppose that the signed
coordinates and the binary encoding of $t$ use at most $S$ bits.  Then $|\alpha|\le1$, and there
is a deterministic polynomial-time exact sparse relation for $y=\alpha x$ with coefficients in a
fixed finite subset of $K_8$.  For each source value $x$, the zero-residual history is unique and
has endpoint exactly $y=\alpha x$.  For any history vector $h$ with endpoint $y$ and residual
$g$, the Euclidean norms satisfy
\[
|y-\alpha x|\le20(S+1)\lVert g\rVert,\qquad
\lVert h\rVert\le100(S+1)\bigl(|x|+\lVert g\rVert\bigr).
\]
The succinct description and access circuits have size polynomial in $S$, independently of
the numerical value of $t$.
\end{lemma}
\begin{proof}
Appendix~\ref{app:stable-quadratic-dyadic-multiplier} gives the addressed construction and its
endpoint and history estimates.
\end{proof}

The lemma applies to entries of normalized supplied terms or projector constraints.  A fixed outer
factor, such as the factor $\sqrt2$ in a constraint residual, is applied after compiling the
normalized entry.  Arbitrary-magnitude succinct scalars lie outside this input domain.

\begin{lemma}[Kernel-preserving Hermitian completion]
\label{lem:cokernel-safe-hermitian-completion}
Let $R_x:X_x\to Y_x$ be a linear map between finite-dimensional Hilbert spaces.  Suppose that, for
positive integer polynomials $L$ and $P$,
\[
\lVert R_x\rVert\le L(S),\qquad
\lVert v\rVert\le P(S)\lVert R_xv\rVert
\quad\text{on every promised NO input}.
\]
Set
\[
K_{R_x}=\begin{pmatrix}0&R_x^\dagger\\ R_x&\Id_{Y_x}\end{pmatrix},
\qquad B(S)=2L(S)+2.
\]
Then
\[
\ker K_{R_x}=\ker R_x\oplus\{0\},\qquad
\left\lVert\frac{K_{R_x}}{B(S)}\right\rVert\le\frac12,
\]
and on a NO input
\[
\sigma_{\min}\!\left(\frac{K_{R_x}}{B(S)}\right)
\ge\frac1{P(S)^2B(S)^2}.
\]
If valid relation objects have unique fixed-width labels, extending this normalized block by the
identity on every invalid label preserves its kernel and NO singular gap.
\end{lemma}
\begin{proof}
The block calculation and the exact label/access realization appear in
Appendix~\ref{app:cokernel-safe-completion}.
\end{proof}

Uniform padding tensors the valid matrix with a spectator identity register and thereby
repeats its singular values and valid nullity.

For the input models below, \emph{clean exact access} uses canonical fixed-width labels and
reversible row, column, exact-value, and adjoint-value circuits that return every work register
to zero. Contributions to the same matrix position are aggregated exactly in $K_8$, and an
entry is suppressed precisely when all four field coordinates vanish.

\begin{definition}[Uniform exact sparse Hermitian access]
\label{def:generic-exact-sparse-access}
A uniform exact sparse Hermitian access scheme is specified by fixed deterministic decoding and
construction algorithms defined on every input string.  On an input $x$ of physical length $S$,
it specifies a Hermitian matrix $A_x$ on $m$ qubits with clean exact access. Each
orientation has at most $d$ valid slots and lists every nonzero entry exactly once and no zero
entry. The row, column, and value circuits use the same validity predicate.

Each logical object has one canonical label. All noncanonical strings form an identity block,
and malformed inputs are mapped to a fixed identity matrix.
\end{definition}

\begin{definition}[Admissible exact-sparse Hermitian input model]
\label{def:generic-admissible-exact-sparse}
A fixed access scheme of Definition~\ref{def:generic-exact-sparse-access} is admissible if a
positive nondecreasing integer polynomial $P_{\mathcal E}$, fixed independently of the input,
bounds $m$, the sparsity, every circuit size and runtime, all register and clean-work widths, and
all exact-arithmetic widths, with $S\le m\le P_{\mathcal E}(S)$.  Every matrix also satisfies
\[
\lVert A_x\rVert\le1.
\]
Its entries have one positive common denominator $h$ such that
\[
(A_x)_{ij}=\frac1h\sum_{r=0}^3z_{ij,r}\zeta_8^r,
\qquad 1\le h\le P_{\mathcal E}(S),
\qquad |z_{ij,r}|\le P_{\mathcal E}(S).
\]
The input also determines a positive reduced rational $\epsilon_x$ with numerator and denominator
at most $P_{\mathcal E}(S)$ and $\epsilon_x\ge P_{\mathcal E}(S)^{-1}$. The coefficient field and
gate set are fixed independently of $x$.
\end{definition}

\begin{definition}[Exact Hermitian nullspace problem]
\label{def:generic-exact-zero}
For a fixed admissible model $\mathcal E$, exactly one of
\[
\textsf{YES}:\ker A_x\ne\{0\},
\qquad
\textsf{NO}:A_x^2\succeq\epsilon_x^2\Id
\]
is promised.  The problem asks which alternative holds.
\end{definition}

\begin{theorem}[Perfect-completeness verification of exact sparse Hermitian nullspaces]
\label{thm:direct-hermitian-nullspace}
For every fixed admissible model $\mathcal E$, the problem of
Definition~\ref{def:generic-exact-zero} lies in $\QMAone$.
\end{theorem}
\begin{proof}
For a fixed admissible model $\mathcal E$, the matrix $A_x$ is Hermitian, has norm at most one,
clean exact row, column, value, and adjoint access, polynomial sparsity, coefficients in $K_8$,
one numerical-polynomial common denominator and height, exact singularity in YES, and least
singular value at least $\epsilon_x$ in NO.  These are the hypotheses of Rudolph's exact sparse
Hermitian verification theorem at cyclotomic level three~\cite{Rudolph2024gateset}.  Its rejection
probability is
\[
c(m)\lVert A_x\ket\psi\rVert^2
\]
for a positive inverse-polynomial $c$.  A kernel vector therefore has zero rejection, while
every NO witness has inverse-polynomial rejection.  At level three the gate matrices are
literally those of $\mathcal G$; exact amplification preserves the zero-rejection witness.
\end{proof}

\begin{definition}[Admissible exact common-kernel data]
\label{def:common-kernel-data}
Fix $K_8=\mathbb Q(\zeta_8)$ and the gate set $\mathcal G$.  A fixed uniform model supplies
polynomially many maps $B_j:X_x\to Y_{j,x}$ and their stack
$R_x:X_x\to\bigoplus_jY_{j,x}$. It has a deterministic decoding algorithm defined on every
input string and clean exact access. There is a positive
nondecreasing integer polynomial $P_{\rm ck}$, fixed by the model, that bounds the sparsity, label
widths, runtimes, workspaces, and circuit sizes.  All coefficients of $R_x$ share one positive
common denominator; this denominator and every coordinate height are at most $P_{\rm ck}(S)$.

The input determines reduced positive rationals $\ell$ and $\delta_x$, whose numerators and
denominators are at most $P_{\rm ck}(S)$, such that $\lVert R_x\rVert\le\ell$ and
$\delta_x\ge P_{\rm ck}(S)^{-1}$.  It is promised that either
\[
 \bigcap_j\ker B_j\ne\{0\}
 \quad\text{or}\quad
 R_x^\dagger R_x=\sum_jB_j^\dagger B_j\succeq\delta_x\Id.
\]
Malformed inputs and invalid labels are assigned fixed identity blocks.
\end{definition}

\begin{theorem}[Exact sparse common-kernel verification]
\label{thm:common-kernel-exact-sparse-transfer}
The common-kernel problem for every model in
Definition~\ref{def:common-kernel-data} lies in $\QMAone$.  More precisely, let
$L=\lceil\ell\rceil$, $\beta=2L+2$, and
\[
 K_R=\begin{pmatrix}0&R_x^\dagger\\R_x&\Id\end{pmatrix}.
\]
The valid block $K_R/\beta$ has norm at most $1/2$, kernel
$\ker R_x\oplus\{0\}$, and NO-instance singular gap at least
$\min(\delta_x,1)/\beta^2$.  Direct-summing the identity on invalid labels preserves the kernel
and gap while giving norm at most one.
\end{theorem}
\begin{proof}
Stacking gives $R_x^\dagger R_x=\sum_jB_j^\dagger B_j$ and
$\ker R_x=\bigcap_j\ker B_j$.  The kernel-preserving completion, its exact spectrum, the integerized
normalization, and the clean access construction are proved in
Appendix~\ref{app:common-kernel-exact-sparse-transfer}.  The resulting matrix satisfies
Theorem~\ref{thm:direct-hermitian-nullspace}.
\end{proof}

\begin{corollary}[Exact sparse supercharge Laplacians]
\label{cor:nilpotent-exact-sparse-transfer}
Suppose a fixed uniform model specifies $\mathcal Q_x$ over $K_8$ with the clean exact access,
aggregation, numerical-height, and polynomial-resource bounds of
Theorem~\ref{thm:common-kernel-exact-sparse-transfer}.  If
\[
\Delta_x=\mathcal Q_x\mathcal Q_x^\dagger+\mathcal Q_x^\dagger\mathcal Q_x,
\]
and the promise is a nonzero kernel of $\Delta_x$ versus
$\Delta_x\succeq\delta_x\Id$ for an inverse-polynomial reduced rational $\delta_x$ with
numerically polynomial numerator and denominator, then the exact-zero problem lies in $\QMAone$.
\end{corollary}
\begin{proof}
Take $B_1=\mathcal Q_x$ and $B_2=\mathcal Q_x^\dagger$.  Their stack satisfies
\[
R_x^\dagger R_x
=\mathcal Q_x^\dagger\mathcal Q_x+\mathcal Q_x\mathcal Q_x^\dagger
=\Delta_x,
\qquad
\ker R_x=\ker\Delta_x.
\]
The kernel identity follows from positivity of both summands.  The supplied row, column,
exact-value, and adjoint-value routines give both residual orientations, so
Theorem~\ref{thm:common-kernel-exact-sparse-transfer} supplies the normalization and verifier.
The Laplacian identity suffices for this application.  A cohomological interpretation of
$\ker\Delta_x$ also requires $\mathcal Q_x^2=0$, as in Section~\ref{sec:susy-types}.
\end{proof}

\section{Minimal supersymmetry on a chain}
\label{sec:n1}

The reduction uses the following spectral map, which is independent of the interaction graph.

\subsection{Spectral mapping from frustration-free Hamiltonians}

\begin{lemma}[Minimal-supersymmetry spectral mapping]
\label{lem:spectral}
Let $h_1,\ldots,h_M\succeq0$, set
\[
 H_0=\sum_jh_j,\qquad
 Q=\sum_j\sqrt{h_j}\otimes\Gamma_j,\qquad
 W=\sum_j\lVert h_j\rVert,
\]
and let the Majoranas act on an ancillary Fock space $\Fcal_{anc}$ and obey
$\{\Gamma_j,\Gamma_k\}=2\delta_{jk}\Id$.  If $W>0$, then
\[
 \ker Q=\ker H_0\otimes\Fcal_{anc},\qquad
 E_0(Q^2)\ge E_0(H_0)^2/W.
\]
If $W=0$, then $H_0=Q=0$ and both kernels are the whole space.
\end{lemma}
\begin{proof}
For a unit vector $\psi$, put
$g_j=\langle\psi|\sqrt{h_j}\otimes\Id|\psi\rangle\ge0$.
Anticommuting $Q$ with normalized real linear combinations of the Majoranas gives
$\lVert g\rVert_2\le\lVert Q\psi\rVert$.  Positivity of the $h_j$ then yields
\[
 \langle\psi|H_0\otimes\Id|\psi\rangle
 \le\sqrt W\,\lVert g\rVert_2\le\sqrt W\,\lVert Q\psi\rVert.
\]
The energy bound follows by minimizing over $\psi$.  If $Q\psi=0$, every $g_j$ vanishes,
and positivity forces every $\sqrt{h_j}$ to annihilate $\psi$, proving the kernel identity.
Appendix~\ref{app:weighted-spectral-mapping} gives the full argument and weighted extension.
\end{proof}

For the chain construction $h_j=\Pi_j$ are nonzero projectors, so
$\sqrt{h_j}=h_j$ and $W=M$.  The supercharge therefore uses the original exact projector
entries, and the NO energy becomes at least $\epsilon_{\mathrm{ff}}^2/M$.

\begin{proposition}[Sharpness of the weighted spectral bound]
\label{prop:quad-tight}
With $H_0,Q,W$ as in Lemma~\ref{lem:spectral}, there is a two-projector family for which
$E_0(Q^2)/E_0(H_0)\to0$, so no universal positive constant gives a linear lower bound
$E_0(Q^2)\ge cE_0(H_0)$. The coefficient one in
$E_0(Q^2)\ge E_0(H_0)^2/W$ is nevertheless optimal: the one-term family
$h_1=\lambda\Id$, with $\lambda>0$, attains equality.
\end{proposition}

\begin{proof}
Both exact witnesses are computed in Appendix~\ref{app:weighted-spectral-mapping}.
\end{proof}

\begin{remark}[Tensor-product interpretation on the chain]
\label{rem:graded}
On the chain, $\sqrt{h_j}\otimes\Gamma_j$ means the Clifford product
$\sqrt{h_j}\Gamma_j$. Lemmas~\ref{lem:jw} and~\ref{lem:sqrt} make $\sqrt{h_j}$ even,
and every ancillary label mode is disjoint from every system support. The even operator therefore
commutes with every label Majorana: each summand is Hermitian, products acquire no Koszul sign,
and $A\otimes\Gamma\mapsto A\Gamma$ on even $A$ is a unital $*$-isomorphism onto its image.
This realizes the Hermitian operator of Lemma~\ref{lem:spectral} in the CAR algebra,
including its kernel identity.
\end{remark}

The support cancellations in Figure~\ref{fig:local-cancellations} explain why the chain
construction below remains local when squared.

\subsection{Exact-zero decision problems}

\begin{definition}[Local Hermitian-square presentation]
\label{def:1glsp-presentation}
A presentation consists of a finite open chain $\Lambda$, a finite mode set $\mathfrak M$, and a
total allocation $\pi:\mathfrak M\to\Lambda$ with at most three modes at each site.  An order on
every fibre $\pi^{-1}(s)$, concatenated in chain order, gives the site-major mode order.  The
presentation contains Hermitian operators $Q_1,\ldots,Q_m$.  Each $Q_a$ satisfies
$\lVert Q_a\rVert\le1$, acts on at most eleven modes, and has site support
$\pi(\operatorname{supp}_{\rm mode}Q_a)$ contained in at most ten consecutive sites.  The
represented Hermitian supercharge is
\[
Q=\sum_{a=1}^m Q_a.
\]
The parity-odd submodel requires every nonzero matrix entry of $Q_a$ to connect
basis states of opposite occupation parity, equivalently $PQ_aP=-Q_a$, for every $a$.
\end{definition}

\begin{definition}[Exact encoding of a local Hermitian-square presentation]
\label{def:1glsp-input}
A well-formed encoding explicitly enumerates every site of $\Lambda$ and every mode of
$\mathfrak M$.  For every site it lists the complete ordered fibre $\pi^{-1}(s)$, including an
empty list when applicable; the fibres must be pairwise disjoint and exhaustive, and their
concatenation is the complete site-major order of the $n_{\mathrm{mode}}$ enumerated modes.  The
map $\pi$ is determined by these lists and is not separately serialized.  For every $Q_a$ the
encoding lists a support consisting only of enumerated mode indices, the declared support and
basis orders, and every nonzero sparse row--column--value entry.  Each coefficient explicitly
records all four integers and the exponent in
$(a+b\sqrt2+ic+id\sqrt2)/2^t$, where $a,b,c,d\in\mathbb Z$ and $t\ge0$; no nonzero entry may be
inferred from symmetry or an omitted block.  Succinct implicit sites, modes, registers, supports,
and gaps between coordinates are not allowed, and every support and chain coordinate refers to an
enumerated index.
\end{definition}

\begin{definition}[Exact-zero Hermitian-square problems]
\label{def:1glsp}
A presentation of Definition~\ref{def:1glsp-presentation}, encoded according to
Definition~\ref{def:1glsp-input}, also supplies a positive unary integer $G_I$. Let $S$ be the
total bit length of this complete canonical encoding, including all $G_I$ unary symbols. Thus
$n_{\mathrm{site}},n_{\mathrm{mode}},m,G_I\le S$. Set $\delta=1/G_I$.
The ungraded problem permits every such input. The intrinsic
minimal-supersymmetry problem is its parity-odd submodel.  In either case the promise is
\[
\textsf{YES}:\ker Q\ne0 \quad\Longleftrightarrow\quad E_0(Q^2)=0,
\qquad
\textsf{NO}:Q^2\succeq\delta\Id.
\]
The question is whether the instance is a YES instance.
\end{definition}

The supplied data are the listed Hermitian terms; a frustration-free or local-positive
decomposition, a square-root list, and a Majorana-label presentation are not input fields.
For $D\ge1$, both problems are defined on a lattice
$\Lambda\subset\mathbb Z^D$, with supports measured in the lattice metric.

\begin{remark}[Termwise representation of an odd supercharge]
\label{rem:1glsp-scope}
The ungraded input permits either parity.  Physical minimal supersymmetry requires the sum $Q$ to
be odd; termwise oddness is the locally checkable representation used here.  Any
globally odd listed operator admits this form after replacing $Q_a$ by
$(Q_a-PQ_aP)/2$, preserving locality, Hermiticity, the coefficient ring, and the norm bound.
\end{remark}

\begin{theorem}[Hardness of minimal supersymmetry on a chain]
\label{thm:minimal-hardness}
The parity-odd intrinsic minimal problem is $\QMAone$-hard. The
reduction outputs the following explicit subclass:
\[
Q=\sum_{c=1}^{M}\Pi_c\Gamma_c,\qquad M=2nK+2,
\]
where the $\Pi_c$ are the even bondwise-merged projectors of Proposition~\ref{prop:ferm} and the
$\Gamma_c$ are distinct anchored Majoranas.  With $G_{\mathrm{gap}}$ from
Proposition~\ref{prop:ambient-raw-gap}, the reduction outputs
\[
G_I:=64M G_{\mathrm{gap}}^2
\]
in unary and satisfies
\[
\text{YES}\Rightarrow E_0(Q^2)=0,
\]
and
\[
\text{NO}\Rightarrow E_0(Q^2)\ge\frac{\epsilon_{\mathrm{ff}}^2}{M}
\ge\frac{1}{64M G_{\mathrm{gap}}^2}=\frac1{G_I}.
\]
Every supercharge summand has arity at most eleven modes. Its system window has width ten,
and $Q^2$ has span at most nineteen encoded sites; its realized span is fifteen.
\end{theorem}
\begin{proof}
First apply Lemma~\ref{lem:ambient-normalization}, and then apply
Proposition~\ref{prop:ferm} to obtain the one-dimensional geometrically
local frustration-free Hamiltonian $H_{in}=\sum_{c=1}^M\Pi_c$, with
$M=2nK+2$, satisfying $E_0(H_{in})=0$ on positive instances and
$E_0(H_{in})\ge\epsilon_{\mathrm{ff}}\ge1/(8G_{\mathrm{gap}})$ on negative
ones. Build the Hermitian supercharge $Q=\sum_c\Pi_c\Gamma_c$ of
Lemma~\ref{lem:spectral} on the chain as in Remark~\ref{rem:graded}, using that
$\sqrt{\Pi_c}=\Pi_c$ for projectors, and set $H=Q^2$. Each product
$\Pi_c\Gamma_c$ has norm one, so the
 normalization of Definition~\ref{def:1glsp-presentation} holds.

The kernel identity and energy bound in Lemma~\ref{lem:spectral} give the stated
YES and NO cases and exclude additional zero modes.  Both $M$ and $G_{\mathrm{gap}}$ are explicitly computable and
polynomially bounded, so the unary output for $G_I$ has polynomial length.

For parity and geometry, each $Q_c=\Pi_c\Gamma_c$ is odd because $\Pi_c$ is even.  Thus the
output satisfies the parity-odd input promise.  Disjoint summands anticommute and contribute no
cross term to $Q^2$.  The site-load bound $D_{\mathrm{site}}\le3$, anchor-multiplicity bound
$A\le2$, and overlap-degree bound $\Delta_{\mathrm{ov}}\le4$ are those of
Proposition~\ref{prop:ferm}.  Anchoring each $\Gamma_c$ in or adjacent
to its projector window then keeps every surviving commutator within nineteen consecutive
encoded sites; the resulting realized span is fifteen. Thus both the local dimension and the interaction
range remain constant.
\end{proof}

\begin{theorem}[Exact-zero completeness for listed Hermitian supercharges]
\label{thm:n1-tiers}
For the exact sparse-list input model of Definition~\ref{def:1glsp}, both the parity-odd
intrinsic minimal-supersymmetry problem and the ungraded Hermitian-square problem are
$\QMAone$-complete.  Hardness already holds for the
parity-odd chain family of Theorem~\ref{thm:minimal-hardness}: each supercharge summand has arity
at most eleven and a system window of width ten, while $Q^2$ has span at most nineteen encoded
sites and realized span fifteen.
\end{theorem}
\begin{proof}
Theorem~\ref{thm:minimal-hardness} gives $\QMAone$-hardness for the parity-odd problem
on the stated chain family.  The same reduction proves hardness of the ungraded problem because
every parity-odd input is also an admissible ungraded input.

For containment, let $x$ be a promised input of physical length $S$.  Proposition~
\ref{prop:n1-hermitian-realization} constructs uniformly an admissible Hermitian exact-sparse
matrix $A_x$ with
\[
\ker A_x\cong\ker Q_x
\]
and gives an inverse-polynomial lower bound on every nonzero singular value of
$A_x$ on a NO input. The common-kernel transfer in
Theorem~\ref{thm:common-kernel-exact-sparse-transfer}, applied by that proposition, supplies a verifier
over $\mathcal G$ with exact acceptance on a kernel witness and inverse-polynomial rejection on
every NO witness.  This proves containment of the ungraded language and of its parity-odd promise
restriction.
\end{proof}

\begin{remark}[Cancellation between local supercharge terms]
\label{rem:intrinsic}
Termwise oddness does not turn the global equation $Q\psi=0$ into the conditions
$Q_a\psi=0$.  On two modes let $Z_2=\Id-2n_2$ and take
\[
Q_1=\Gamma_1,\qquad Q_2=-\Gamma_1Z_2,
\qquad Q=Q_1+Q_2=2\Gamma_1n_2.
\]
Both terms are exact, local, Hermitian, parity odd, and norm one.  If $0\ne\psi$ and
$n_2\psi=0$, then
$Q\psi=0$ while $Q_1\psi=-Q_2\psi\ne0$.  A termwise kernel test can therefore reject a zero
mode.
\end{remark}

The relation in Appendix~\ref{app:n1-hermitian-relation} uses each entry's encoded binary
exponent and aggregates all occurrences at each global target, preserving ordered multiplicity
and exact cancellation.  Its common-denominator, height, normalization, and clean-access bounds
apply to the constructed verifier matrix $A_x$, which has no geometric-locality requirement.

\section{Zero modes of listed nilpotent supercharges}
\label{sec:explicit-n2}

For $m$ ordered modes, the global occupation basis and CAR action are
\[
\ket x=(c_0^\dagger)^{x_0}\cdots(c_{m-1}^\dagger)^{x_{m-1}}\ket\Omega,
\]
\[
c_j\ket x=(-1)^{\sum_{k<j}x_k}x_j\ket{x-e_j},\qquad
c_j^\dagger\ket x=(-1)^{\sum_{k<j}x_k}(1-x_j)\ket{x+e_j}.
\]

\begin{definition}[Explicit bounded-arity supercharge tables]
\label{def:explicit-table-input}
Fix term arity $r_*=13$ and coefficient field
$K_8=\mathbb Q(\zeta_8)$.  An instance of encoded length $S$
contains only:
\begin{enumerate}
\item a mode count $m$ and declared site-major occupation order; when chain coordinates are
present, they explicitly enumerate the chain sites and give the complete ordered per-site mode
allocation, including empty fibres, whose disjoint exhaustive union is all $m$ modes and whose
concatenation is the declared order; these lists determine the total map $\pi$ of
Definition~\ref{def:mode-allocation-locality};
\item an ordered list $\mathcal Q=\sum_{a=1}^M\mathcal Q_a$;
\item for each term, a strictly increasing support of size at most $r_*$ and its complete exact
local occupation-basis table in that support order;
\item one positive globally reduced common denominator $h$ and, for every entry, four signed
integers $(z_0,z_1,z_2,z_3)$ representing
$h^{-1}\sum_{j=0}^3z_j\zeta_8^j$;
\item one positive reduced rational $\delta=u/v$; and
\item syntactic information specifying support, local-basis, parity, and order conventions,
with the global occupation basis and CAR action fixed above.
\end{enumerate}

The encoding obeys the absolute numerical bounds
\[
m,M\le S,\qquad h,|z_j|,u,v\le S^2,
\qquad \delta\ge S^{-2}.
\]
Every nonzero table entry changes fermion parity, and $\lVert\mathcal Q_a\rVert\le1$.
\end{definition}

\begin{definition}[Exact-zero problem for explicit nilpotent supercharges]
\label{def:explicit-table-exact-zero}
An input of Definition~\ref{def:explicit-table-input} is promised
to satisfy
$\mathcal Q^2=0$ on the full Fock space.  With
$\Delta=\{\mathcal Q,\mathcal Q^\dagger\}$, exactly one of
\[
\textsf{YES}:E_0(\Delta)=0,
\qquad
\textsf{NO}:\Delta\succeq\delta\Id
\]
is promised.  We call this the explicit exact-zero nilpotent-supercharge problem.
\end{definition}

\begin{definition}[Chain-local Hamiltonian restriction]
\label{def:target-chain-restriction}
For inputs carrying the declared chain coordinates, this restriction fixes $R_*=19$.  It promises
that every nonzero monomial in the canonically normal-ordered, exactly aggregated expansion of
$\Delta$ has site support, computed as the $\pi$-image of its mode support, in an interval of at
most $R_*$ consecutive encoded sites.  The locality restriction applies to the Hamiltonian alone;
$\mathcal Q$ and its individual terms retain the general geometry allowed by the input model.
\end{definition}

\begin{proposition}[Exact sparse access from bounded-arity supercharge tables]
\label{prop:target-list-admissible}
The representation and promise of Definitions~\ref{def:explicit-table-input}
and~\ref{def:explicit-table-exact-zero} satisfy the uniform access and
numerical hypotheses of Corollary~\ref{cor:nilpotent-exact-sparse-transfer}.
\end{proposition}
\begin{proof}
For a support $J=(j_1<\cdots<j_r)$, convert its complete local table by a fixed exact integer
change of basis to the normal words
\[
W_J(\alpha,\beta)=
\prod_{t=1}^r(c_{j_t}^\dagger)^{\alpha_t}
\prod_{t=r}^1(c_{j_t})^{\beta_t},
\qquad \alpha,\beta\in\{0,1\}^r.
\]
Because $r\le r_*$, the transform uses only a fixed number of additions and sign changes and
introduces no denominator.  On a column basis state, apply every word from right to left,
computing the global CAR prefix parity at each mode.  This enumerates every column contribution
with the signs from modes between support elements included.  Applying the adjoint words and the
exact involution
$(z_0,z_1,z_2,z_3)\mapsto(z_0,-z_3,-z_2,-z_1)$ supplies the opposite orientation.
At most $M\,4^{r_*}$ entry contributions are generated for $\mathcal Q$ in either orientation.
Sort them by output index and aggregate equal indices exactly, suppressing an entry only when all
four field coordinates vanish. The fixed basis change multiplies coordinate heights by only a constant;
aggregating at most $M\,4^{r_*}$ entries leaves the common denominator $h$ unchanged and gives
polynomial numerical coordinate height. Copying the canonical output entries before reversing
the parity, sorting, and arithmetic work gives clean bidirectional access. The gap, parity, and
term-norm bounds are those of
Definition~\ref{def:explicit-table-input}; full-space nilpotency and the spectral alternative are
promised in Definition~\ref{def:explicit-table-exact-zero}.
\end{proof}

\paragraph{Why two auxiliary modes?}
The one-mode doubling $\mathcal Q=Q\otimes f^\dagger$ is nilpotent and satisfies
$\{\mathcal Q,\mathcal Q^\dagger\}=Q^2\otimes\Id$, but odd $Q$ makes it parity even.
It therefore falls outside the fermionic supercharge convention of~\cite{Cade2024}.
The even two-mode operator below preserves the odd parity of $Q$ and the locality
of its Hamiltonian.  Its auxiliary modes remain shared by all supercharge terms.

\begin{theorem}[Nilpotent extension with two auxiliary modes]
\label{thm:two-mode-nilpotentization}
Let $Q=Q^\dagger$ act on a $d$-dimensional system space, and let $a<b$ be two fresh fermionic
modes.  Define $s=ba+a^\dagger b$ and $\mathcal Q=Qs$.  On the full two-mode Fock space,
\[
 s^2=0,\qquad ss^\dagger+s^\dagger s=\Id_{ab},\qquad \lVert s\rVert=1,
\]
and hence
\[
 \mathcal Q^2=0,\qquad
 \{\mathcal Q,\mathcal Q^\dagger\}=Q^2\otimes\Id_{ab}.
\]
If $k=\dim\ker Q$, the Hamiltonian repeats every eigenvalue of $Q^2$ four times, has
zero-energy-space dimension $4k$, and has the same least positive eigenvalue whenever the positive
spectrum is nonempty.  Its kernel and cohomology dimensions are
\[
 \dim\ker\mathcal Q=2d+2k,\qquad \dim H(\mathcal Q)=4k.
\]
\end{theorem}
\begin{proof}
The auxiliary CAR calculation, the spectral and dimension counts, and the locality boundary are
proved in Appendix~\ref{app:two-mode-nilpotentization}.
\end{proof}

Multiplication by $s$ preserves the coefficient field and term norms and adds at most
two modes to each support.  Figure~\ref{fig:two-mode-mechanism} displays its action.
Even when $Q$ is invertible, $\ker\mathcal Q$ has dimension $2d$; those vectors alone
do not signal zero energy.  The zero-energy space is the common kernel of
$\mathcal Q$ and $\mathcal Q^\dagger$, which represents $H(\mathcal Q)$ and has dimension $4k$.

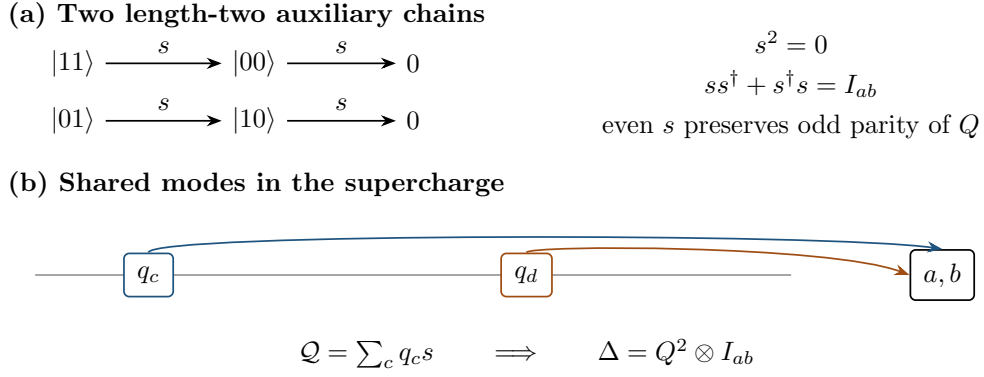
\begin{figure}[htbp]
\centering
\begin{tikzpicture}[x=1cm,y=.80cm,font=\small]
  \node[anchor=west,font=\small\bfseries] at (0,0) {(a) Two length-two auxiliary chains};
  \node (a) at (1,-.8) {$\ket{11}$};
  \node (b) at (3.4,-.8) {$\ket{00}$};
  \node (c) at (5.5,-.8) {$0$};
  \draw[diagram arrow] (a)--node[above] {$s$}(b);
  \draw[diagram arrow] (b)--node[above] {$s$}(c);
  \node (d) at (1,-1.75) {$\ket{01}$};
  \node (e) at (3.4,-1.75) {$\ket{10}$};
  \node (f) at (5.5,-1.75) {$0$};
  \draw[diagram arrow] (d)--node[above] {$s$}(e);
  \draw[diagram arrow] (e)--node[above] {$s$}(f);
  \node[align=center] at (10.5,-1.15)
    {$s^2=0$\\[3pt]$ss^\dagger+s^\dagger s=I_{ab}$\\[3pt]even $s$ preserves odd parity of $Q$};
  \node[anchor=west,font=\small\bfseries] at (0,-2.8) {(b) Shared modes in the supercharge};
  \draw[black!40] (.5,-4.3)--(10.5,-4.3);
  \node[diagram box,draw=diagramblue,fill=white] (q1) at (2,-4.3) {$q_c$};
  \node[diagram box,draw=diagramorange,fill=white] (q2) at (7,-4.3) {$q_d$};
  \node[diagram box] (ab) at (12.5,-4.3) {$a,b$};
  \draw[diagram arrow,diagramblue] (q1.north) .. controls (2,-3.6) and (11,-3.6) .. (ab.north);
  \draw[diagram arrow,diagramorange] (q2.north) .. controls (7,-3.8) and (11,-3.8) .. (ab.west);
  \node at (7,-5.6) {$\mathcal Q=\sum_cq_cs\qquad\Longrightarrow\qquad
    \Delta=Q^2\otimes I_{ab}$};
\end{tikzpicture}
\caption{The parity-preserving nilpotent extension.  (a) In the ordered basis
$\ket{n_an_b}=(a^\dagger)^{n_a}(b^\dagger)^{n_b}\ket{00}$, the two nonzero arrows
of $s=ba+a^\dagger b$ have amplitude one.  Their domain and image are complementary
two-dimensional spaces, giving the two identities on the right.  The terminal zero
denotes the zero vector.  (b) All supercharge terms share $a,b$, so their chain spans
need not be bounded.  In the Laplacian the auxiliary factor is the identity, leaving
the interaction range of $Q^2$ unchanged and repeating its spectrum four times.}
\label{fig:two-mode-mechanism}
\end{figure}

\begin{lemma}[Two-mode extension of the chain construction]
\label{lem:nilpotent-two-mode-extension}
Let $Q_{\min}=\sum_cq_c$, with $q_c=\Pi_c\Gamma_c$, be the Hermitian chain supercharge produced
by Theorem~\ref{thm:minimal-hardness}.  With the two fresh modes $a<b$ and the operator $s$ of
Theorem~\ref{thm:two-mode-nilpotentization}, the terms $\mathcal Q_c=q_cs$ are odd, have norm at
most one and arity at most thirteen, and sum to a nilpotent supercharge on the full Fock space.
Its Hamiltonian $Q_{\min}^2\otimes\Id_{ab}$ has the same chain range and realized span as
$Q_{\min}^2$.  The supercharge terms are generally non-Hermitian and need not be geometrically local.
\end{lemma}
\begin{proof}
Each $\Pi_c$ is an even bondwise-merged projector and $\Gamma_c$ its distinct anchored Majorana,
so $q_c$ is odd and Hermitian with norm at most one.  The even operator $s$ preserves this parity
and norm.  The support consists of at most ten system modes, one label mode, and the two
auxiliary modes.  Multiplication by $s$ duplicates local entries with signs, preserving the
$K_8$ field, common denominator, and coordinate-height bounds.  The full-space identities and
Hamiltonian locality follow from Theorem~\ref{thm:two-mode-nilpotentization}.
\end{proof}

\begin{proposition}[Hardness with a chain-local Hamiltonian]
\label{prop:explicit-n2-hardness-transfer}
The fixed $d=11$ one-dimensional quantum-SAT family over $\mathcal G$ of
Theorem~\ref{thm:rgn} and Lemma~\ref{lem:source-normalization} admits a deterministic
polynomial-time reduction to the explicit exact-zero nilpotent-supercharge problem of
Definition~\ref{def:explicit-table-exact-zero}, with the chain-local restriction of
Definition~\ref{def:target-chain-restriction}.
Its outputs have supercharge-term arity at most $13$, globally reduced common denominator at most
$8$, integral-basis coordinate height at most $72$, and Hamiltonian range $R_*=19$ and realized
span at most $15$.  The supercharge terms in this output family are generally
non-Hermitian.
\end{proposition}
\begin{proof}
Apply Lemma~\ref{lem:nilpotent-two-mode-extension} to the Hermitian chain supercharge
of Theorem~\ref{thm:minimal-hardness}.

Appendix~\ref{app:explicit-n2-reduction} constructs the exact CAR tables in polynomial time
and verifies the promise transfer and the stated numerical and geometric bounds.
\end{proof}

\begin{theorem}[Nilpotent-supercharge completeness]
\label{thm:explicit-n2-completeness}
The problem in Definition~\ref{def:explicit-table-exact-zero} is $\QMAone$-complete.  Hardness
already holds on the subfamily of
Definition~\ref{def:target-chain-restriction}, whose Hamiltonian has range nineteen and realized
span at most fifteen on a chain.
\end{theorem}
\begin{proof}
Proposition~\ref{prop:target-list-admissible} and
Corollary~\ref{cor:nilpotent-exact-sparse-transfer} give $\QMAone$ containment for the general
explicit nilpotent-supercharge problem.  Proposition~\ref{prop:explicit-n2-hardness-transfer}
gives $\QMAone$-hardness on the stated chain-local Hamiltonian subfamily.
\end{proof}


\section{Discussion}
\label{sec:dim}

The two exact input models give $\QMAone$-complete zero-energy problems whose hard Hamiltonians
have fixed local dimension and interaction range on a chain.

Embedding a chain as one row of a higher-dimensional lattice preserves the supercharge, its
spectrum, and its interaction range.  The hardness results consequently extend to higher
dimensions with the same local parameters.

The Hermitian construction has a quadratic soundness bound.  Proposition~\ref{prop:quad-tight}
shows that this behaviour is intrinsic to the single-Majorana spectral map, up to constants.  A
linear bound would require a different local Hermitian construction.  Lowering the interaction
arity while preserving exact zero energy is another open direction.

The general ground-energy problem with a positive YES threshold remains open under geometric
locality.  For nilpotent supercharges, requiring the individual supercharge terms as well as
their Hamiltonian to be local on a chain is a further geometric problem.

\appendix
\section{Verification of sparse linear constraints}
\label{app:exact-sparse-transfer}

\subsection{Exact reversible access circuits}
The reversible Boolean subroutines used below are synthesized exactly over $\mathcal G$.  The
basic identities are
\[
X=\widehat HT^4\widehat H,
\qquad T^\dagger=T^7.
\]
For controls $a,b$ and target $t$, use the gate sequence
\[
\begin{split}
&\widehat H_t,\ CX_{b,t},\ T_t^\dagger,\ CX_{a,t},\ T_t,\ CX_{b,t},
T_t^\dagger,\ CX_{a,t},\\
&T_b,\ T_t,\ \widehat H_t,\ CX_{a,b},\ T_a,\ T_b^\dagger,\ CX_{a,b}.
\end{split}
\]
In the ordered basis $\ket{000},\ket{001},\ldots,\ket{111}$, direct multiplication gives
\[
\sum_{x\in\{0,1\}^3\setminus\{110,111\}}\ket x\!\bra x
+\ket{110}\!\bra{111}+\ket{111}\!\bra{110}.
\]
This is the phase-sensitive Toffoli matrix.

Replacing NOT, XOR, copy, and AND by $X$, CNOT, and the displayed Toffoli makes the required
Boolean computations exactly reversible.  Fixed sorting networks and
compute--copy--uncompute give clean outputs with polynomial overhead.

\subsection{Stable multiplication by quadratic-dyadic coefficients}
\label{app:stable-quadratic-dyadic-multiplier}

We prove Lemma~\ref{lem:stable-quadratic-dyadic-multiplier}.  The case $\alpha=0$ is represented
by the absence of a coefficient incidence, so assume $\alpha\ne0$.  Put
\[
\beta=1+\sqrt2,\qquad \rho=\sqrt2-1=\beta^{-1},\qquad k=S+2,
\]
and let $T(u,v)=(u+2v,u+v)$.  Compute
\[
(A,B)=T^k(a,b),\qquad (C,D)=T^k(c,d).
\]
Thus
\[
A+B\sqrt2=\beta^k(a+b\sqrt2),\qquad
C+D\sqrt2=\beta^k(c+d\sqrt2).
\]

For $z=a+b\sqrt2\ne0$, let $z'=a-b\sqrt2$ and
$H=\max\{|a|,|b|,1\}<2^S$.  Since $zz'=a^2-2b^2$ is a nonzero integer,
\[
|z'|<3H,\qquad |z|>\frac1{3H},\qquad \frac{|z'|}{|z|}<9\,4^S.
\]
The conjugate of $\beta$ is $-\rho$, and therefore
\[
\frac{|(-\rho)^kz'|}{|\beta^kz|}
=\beta^{-2k}\frac{|z'|}{|z|}<1.
\]
Writing $w=\beta^kz=A+B\sqrt2$ and $w'=(-\rho)^kz'=A-B\sqrt2$ shows from
\[
A=\frac{w+w'}2,
\qquad
B=\frac{w-w'}{2\sqrt2}
\]
that $A$ and $B$ are nonzero and have the same sign.  The same argument applies to $C,D$;
an originally zero pair remains zero.

Let
\[
m=\max\{\operatorname{bitlen}|A|,\operatorname{bitlen}|B|,
          \operatorname{bitlen}|C|,\operatorname{bitlen}|D|\}.
\]
One application of $T$ increases the largest absolute coordinate by at most a factor three, so
$1\le m\le3S+5$.  For $0\le j<m$, define
\[
d_j=\operatorname{sgn}(A)A_j+
\operatorname{sgn}(B)B_j\sqrt2+
i\bigl(\operatorname{sgn}(C)C_j+
\operatorname{sgn}(D)D_j\sqrt2\bigr),
\]
where $N_j$ is bit $j$ of $|N|$.  Every coordinate of $d_j$ lies in
$\{-1,0,1\}$.  With $s_0=0$ fixed rather than represented by a variable, impose
\[
r^H_j=s_{j+1}-\frac{s_j+d_jx}{2},
\qquad 0\le j<m.
\]
At zero residual,
\[
s_m=h x,
\qquad
h=\frac{A+B\sqrt2+i(C+D\sqrt2)}{2^m}.
\]
Sign alignment prevents cancellation in each real quadratic component.  Hence every exact partial
Horner multiplier has modulus below four, $|h|\ge1/2$, and the inverse of the lower-bidiagonal
residual matrix has norm at most two.  In particular,
\[
\lVert s-s^{(0)}(x)\rVert\le2\lVert r_H\rVert,
\qquad
\lVert s\rVert\le4\sqrt m\,|x|+2\lVert r_H\rVert.
\]

The equality
\[
\alpha=h\gamma,
\qquad
\gamma=2^{m-t}\rho^k
\]
and the bounds $|\alpha|\le1$, $|h|\ge1/2$ give $0<\gamma\le2$.  Put $e=m-t$.
If $e<0$, use $k$ factors $\rho$ followed by $t-m$ factors $1/2$.  If $e\ge0$,
interleave the $e$ factors $2$ with the $k$ factors $\rho$: starting with prefix product $p=1$,
insert a factor $\rho$ before the next doubling whenever $p>1$, append the factor $2$, and place
all unused $\rho$ factors at the end.

The \emph{multiplier guard} accepts the $e<0$ branch and, in the $e\ge0$ branch, checks that
a factor $\rho$ remains available whenever one is requested.  If none remains, it rejects and
stops the construction; otherwise it accepts after the final doubling.  This check performs
at most $e\le m\le3S+5$ doublings and $k=S+2$ contractions.  Exact comparison of $p$ with one
uses $\mathbb Z[\sqrt2]$ coordinates of $O(S)$ bits.  The $e<0$ branch represents its
half-contraction tail by a binary length and does not iterate through it.
For a normalized input coefficient, a requested contraction is always available: if none
remained while $p>1$, the next doubling and all remaining factors would give $\gamma>2$,
contradicting $|\alpha|\le1$ and $|h|\ge1/2$.  Thus the guard accepts every coefficient
covered by Lemma~\ref{lem:stable-quadratic-dyadic-multiplier}.
In the $e\ge0$ branch of an accepted construction, the active prefixes satisfy
\[
\rho<p\le2,
\]
so every contiguous factor product is less than $2/\rho<5$.  The explicitly active word has
length at most $4S+7$; a possibly much longer half-contraction tail is addressed by its binary
stage label.

For either factor word $f_1,\ldots,f_L$, set $q_0=s_m$, treating it as an alias rather than a
second variable, and impose
\[
r^F_j=q_j-f_jq_{j-1},
\qquad 1\le j\le L,
\qquad y=q_L.
\]
Zero residual gives $y=\alpha x$ and determines every history variable uniquely.  The geometric
inverse bound in the $e<0$ case and the contiguous-product bound in the $e\ge0$ case
give the safe uniform estimates
\[
|y-\alpha x|\le20(S+1)\lVert(r_H,r_F)\rVert,
\]
\[
\lVert(s,q)\rVert
\le100(S+1)\bigl(|x|+\lVert(r_H,r_F)\rVert\bigr).
\]

All row coefficients belong to a fixed finite subset of $K_8$: signs of
\[
1,\quad2,\quad\rho,\quad\frac12,\quad
\frac{p+q\sqrt2+i(r+s\sqrt2)}2,
\qquad p,q,r,s\in\{-1,0,1\}.
\]
A Horner row has at most three nonzero entries and a factor row at most two.  A source-column
query scans only the $m\le3S+5$ digit positions; a history column has only adjacent rows.  In the
long half-contraction tail, predecessor, successor, and endpoint arithmetic use the binary stage
label and never iterate through the numerical value of $t$.  Exact zero suppression and
compute--copy--uncompute give independent clean row and column access.  This proves
Lemma~\ref{lem:stable-quadratic-dyadic-multiplier}.

\subsection{Kernel-preserving Hermitian completion and normalization}
\label{app:cokernel-safe-completion}

Let $R=R_x$.  From
\[
K_R(v,w)=(R^\dagger w,Rv+w),
\]
the equation $K_R(v,w)=0$ gives $w=-Rv$ and $R^\dagger Rv=0$.  Taking the inner product with
$v$ gives $Rv=0$, hence $w=0$ and
\[
\ker K_R=\ker R\oplus\{0\}
\]
with equal nullity.  For each nonzero singular value $\sigma$ of $R$, the corresponding
two-dimensional block of $K_R$ has eigenvalues
\[
\lambda_\pm(\sigma)=\frac{1\pm\sqrt{1+4\sigma^2}}2.
\]
Every direction in $\ker R^\dagger$ has eigenvalue one.  If $\lVert R\rVert\le L$ and
$B=2L+2$, then $\lVert K_R/B\rVert\le1/2$ and
\[
\left|\frac{\lambda_-(\sigma)}B\right|
=\frac{2\sigma^2}{B(\sqrt{1+4\sigma^2}+1)}
\ge\frac{\sigma^2}{B^2}.
\]
Thus a NO lower bound $\sigma\ge P^{-1}$ transfers to
$\sigma_{\min}(K_R/B)\ge(P^2B^2)^{-1}$.  The $1/B$ eigenvalues from
$\ker R^\dagger$ are no smaller.

The label grammar is application-specific, but the access rule is common.  Every domain variable
and codomain residual has one tag and one fixed-width tuple of zero-padded fields.  Range,
incidence, stage, and redundant-field checks are recomputed from the input.  An alias is owned by
one label: for example, $s_0=0$ is not a variable, $q_0$ uses the label of $s_m$, and the endpoint
$y$ uses the label of $q_L$.  Alternative alias strings are invalid; otherwise equality-only
gauge directions would be introduced.

Row access to the domain half of $K_R$ invokes column access to $R$ and conjugates the returned
values.  Row access to the codomain half invokes row access to $R$ and adds its diagonal identity
entry.  Column access uses the opposite orientations.  Exact conjugation in the chosen basis is
\[
(z_0,z_1,z_2,z_3)\longmapsto(z_0,-z_3,-z_2,-z_1).
\]
Each enumerator generates a polynomial list of entry contributions, aggregates identical
positions, suppresses an aggregate exactly when all four coordinates vanish, and assigns
canonical slots by
reversible prefix counting.  An invalid row or column lists only its diagonal identity, and the
value circuit uses the same predicates and block rules.

If the relation has one numerical-polynomial common denominator $h_R$ and coordinate height
$H_R$, then the completed and normalized matrix has common denominator $Bh_R$.  An
$R/B$ entry retains the numerator coordinates of $R$, a codomain identity entry $1/B$ has
numerator $h_R$, and an invalid identity entry has numerator $Bh_R$.  Thus the common denominator
and coordinate height are bounded by $Bh_R$ and $\max\{H_R,Bh_R\}$, respectively.
These bounds meet Definition~\ref{def:generic-admissible-exact-sparse} without computing a
globally reduced denominator.  This proves
Lemma~\ref{lem:cokernel-safe-hermitian-completion}.

\subsection{Common-kernel reduction}
\label{app:common-kernel-exact-sparse-transfer}

Let $Y_x=\bigoplus_jY_{j,x}$. Stacking the residuals gives, without approximation,
\[
\lVert R_xv\rVert^2=\sum_j\lVert B_jv\rVert^2,
\qquad
R_x^\dagger R_x=\sum_jB_j^\dagger B_j,
\qquad
\ker R_x=\bigcap_j\ker B_j.
\]
Lemma~\ref{lem:cokernel-safe-hermitian-completion} gives
$\ker K_R=\ker R_x\oplus\{0\}$.  On a NO instance,
$\sigma_{\min}(R_x)^2\ge\delta_x$. Substituting this bound and
$\beta=2L+2$ into the completion estimate shows that every nonzero singular value of
$K_R/\beta$ is at least $\min(\delta_x,1)/\beta^2$, while
$\lVert K_R/\beta\rVert\le1/2$.

A canonical row label of $R_x$ is a pair $(j,y)$, where $j$ is a fixed-width residual index
and $y$ is a canonical row label for $B_j$; its columns use the labels of $X_x$. Row access
decodes $j$ and calls the corresponding residual routine. Column access loops over the
polynomially many residual and slot indices. All residual routines are fixed by the model and
belong to the same uniform access algorithm.  Apply the exact aggregation, zero suppression,
and block-access construction of Appendix~\ref{app:cokernel-safe-completion}, using the supplied
opposite-direction and adjoint-value routines for $R_x^\dagger$.  Copy the declared output and
reverse the decoding, selection, arithmetic, and sorting workspace to obtain clean access.

The block flag selects $R_x$, $R_x^\dagger$, or the $Y_x$ identity. Invalid square labels
select their own identity entry, and malformed inputs select a fixed identity matrix. Because
$\beta=2L+2$ is a numerically polynomial integer, normalizing by $\beta$ multiplies the common
coefficient denominator by only a numerical-polynomial factor. If
$\bar\delta_x=\min(\delta_x,1)$, reducing
$\bar\delta_x/\beta^2$ cannot increase its numerator or denominator beyond the declared
numerical-polynomial bounds. Thus the Hermitian matrix, including its identity blocks on invalid
labels, meets the hypotheses of
Theorem~\ref{thm:direct-hermitian-nullspace}, which proves
Theorem~\ref{thm:common-kernel-exact-sparse-transfer}.

\section{Normalization and gap of the one-dimensional history Hamiltonian}
\label{app:raw-source-gap}

This appendix proves the quantitative raw-source bound used, after bondwise merging, in
Theorem~\ref{thm:minimal-hardness}.

\begin{proof}[Proof of Lemma~\ref{lem:ambient-normalization}]
Fix the verifier family for the language. By its soundness promise, there is a computable integer
polynomial $P_{\mathcal L}(s)\ge1$ such that the acceptance effect $A_x$ on a NO input of length
$s$ obeys
\[
\lVert A_x\rVert\le1-\frac1{P_{\mathcal L}(s)}.
\]
Set $t=P_{\mathcal L}(s)$, run $t$ copies on disjoint witness and clean-ancilla blocks, and
compute the coherent AND of their output bits. The acceptance effect on the joint witness is
$A_x^{\otimes t}$. This operator identity also covers witnesses entangled across the blocks, and
\[
\lVert A_x^{\otimes t}\rVert=\lVert A_x\rVert^t
\le(1-1/t)^t\le\frac12.
\]
A tensor product of perfectly accepted witnesses retains acceptance one.  The phase-sensitive
Toffoli decomposition in Appendix~\ref{app:exact-sparse-transfer} implements the coherent AND
exactly with polynomial circuit size over $\mathcal G$.
\end{proof}

\subsection{Local definition of the reconstructed clock}
\label{app:clock-constructor}

Fix the padded nearest-neighbor circuit with $n,K\ge2$, gate table
$U_{r,j}$ for $0\le r<K$ and $0\le j<n-1$, and zero-initialized wire set
$\mathcal A\subseteq\{0,\ldots,n-1\}$.  The output is wire $n-1$, with
acceptance value $1$.  Put $N=2nK+1$ and number the sites from $0$ to $N-1$.
The bulk transition patterns are adapted
from~\cite[Eqs.~(4.149)--(4.162)]{Nagaj2008}; the following alphabets, local
predicates and static placements define our reconstruction.

The ordered bases at even and odd sites are, respectively,
\[
\begin{split}
\mathcal B_E&=(A_0,A_1,M_0,M_1,I_0,I_1,O,X,T,B,x),\\
\mathcal B_O&=(A_0,A_1,M_0,M_1,I_0,I_1,O,X,T,B,o).
\end{split}
\]
The labels $A,M,I$ carry one qubit; their subscripts give its computational
basis.  The other labels are one-dimensional.  In either ordered basis,
encode the eleven states by the five-bit binary representations of
\[
0,3,5,6,9,10,12,15,17,18,20,
\]
in that order.  These eleven words all have even weight.  Write $J_p$ for the
projector onto their $21$-dimensional complement at parity $p\in\{E,O\}$.
Every operator defined on valid labels below is extended by zero outside
this code, unless an identity on all $32$ states is specified.

Let $\mathcal E$ and $\mathcal O$ be the even and odd physical bond indices
in $\{0,\ldots,N-2\}$.  Define
\[
\begin{split}
\mathcal B_{\mathrm{end}}&=\{2(r+1)n-1:0\le r<K\},\\
\mathcal B_{\mathrm{act}}&=\{2rn-1:1\le r<K\},\\
\mathcal B_{\mathrm{gate}}&=\{2rn+2j+1:0\le r<K,\ 0\le j<n-1\}.
\end{split}
\]
For each row of Table~\ref{tab:raw-clock-rules}, the displayed map extends
linearly to an isometry $V$ from its domain $\mathcal D$ onto its image
$\mathcal I$.  Internal indices range over $\{0,1\}$ and the gate basis is
ordered as $00,01,10,11$.  The map acts only on the two indicated sites.
Its raw propagation projector is
\begin{equation}
P_V=\frac12(P_{\mathcal D}+P_{\mathcal I}-V-V^\dagger).
\label{eq:raw-clock-projector}
\end{equation}
The two endpoint label sectors are orthogonal, so this is a projector of
rank $\dim\mathcal D$.  Every row is included at every bond satisfying its
predicate, including on configurations outside the intended history.

\begin{table}[htbp]
\centering
\small
\begin{tabular}{c|l|l}
Rule & Channel map & Bond predicate \\ \hline
$1a$ & $\ket{A_p,o}\mapsto\ket{x,A_p}$ & $b\in\mathcal E$ \\
$1b$ & $\ket{M_p,o}\mapsto\ket{x,M_p}$ & $b\in\mathcal E$ \\
$2a$ & $\ket{A_p,I_q}\mapsto\sum_{u,v}(U_{r,j})_{uv,pq}\ket{I_u,A_v}$
 & $b=2rn+2j+1\in\mathcal B_{\mathrm{gate}}$ \\
$2b$ & $\ket{M_p,I_q}\mapsto\ket{I_p,M_q}$ & $b\in\mathcal O$ \\
$2A$ & $\ket{A_p,O}\mapsto\ket{I_p,T}$ & $b\in\mathcal B_{\mathrm{end}}$ \\
$2M$ & $\ket{M_p,O}\mapsto\ket{I_p,T}$ & $b\in\mathcal O\setminus\mathcal B_{\mathrm{end}}$ \\
$3a$ & $\ket{T,O}\mapsto\ket{x,T}$ & $b\in\mathcal E$ \\
$3b$ & $\ket{T,O}\mapsto\ket{B,O}$ & $b\in\mathcal O$ \\
$4a$ & $\ket{x,B}\mapsto\ket{B,o}$ & $b\in\mathcal E$ \\
$4b$ & $\ket{I_p,B}\mapsto\ket{B,I_p}$ & $b\in\mathcal O$ \\
$5a$ & $\ket{X,B}\mapsto\ket{X,T}$ & $b\in\mathcal O$ \\
$5b$ & $\ket{T,o}\mapsto\ket{X,T}$ & $b\in\mathcal E$ \\
$6a$ & $\ket{T,I_p}\mapsto\ket{X,A_p}$ & $b\in\mathcal B_{\mathrm{act}}$ \\
$6b$ & $\ket{T,I_p}\mapsto\ket{X,M_p}$ & $b\in\mathcal O\setminus\mathcal B_{\mathrm{act}}$
\end{tabular}
\caption{The fourteen static transition isometries.  The labels in a ket
are ordered from left to right.}
\label{tab:raw-clock-rules}
\end{table}

At site $0$ also include the endpoint projector $P_{\mathrm{EP}}$
obtained from Eq.~\eqref{eq:raw-clock-projector} with $V=\ket T\!\bra B$.
There are no other transition terms.

\paragraph*{Canonical path and local clock predicates.}
\label{def:canonical-clock}
Suppress the internal bits and start with the word
\[
c_0=A\,o\,(I\,o)^{n-1}O^{N-2n}.
\]
Put $S_{\mathrm{mod}}=(K-1)n$.  For $s=0,\ldots,S_{\mathrm{mod}}$, let
$H_s=A$ when $s\equiv0\pmod n$ and $H_s=M$ otherwise.  At the start of
sweep $s$, this active label is at site $2s$.  Generate successive words
by the following procedure, using the forward directions in the table.
\begin{enumerate}
\item For each $j=0,\ldots,n-1$, apply $1a$ or $1b$ at $b=2s+2j$,
according to $H_s$.  Then apply $2a$ or $2b$ at $b+1$ if $j<n-1$;
at $j=n-1$, apply $2A$ or $2M$ there.  An $A$ sweep uses
$U_{s/n,j}$ at its gate steps.
\item If $s=S_{\mathrm{mod}}$, stop.  Otherwise apply $3a$ at
$2s+2n$ and $3b$ at $2s+2n+1$.
\item Move $B$ left to site $2s$, using $4a$ when its left neighbor is
$x$ and $4b$ when its left neighbor is $I$.
\item Apply $\mathrm{EP}$ if $s=0$, or $5a$ at $2s-1$ if $s>0$.
Then apply $5b$ at $2s$ and the selected $6a$ or $6b$ at $2s+1$.
\end{enumerate}
Every step satisfies its static bond predicate.  The final word is
$X^{2S_{\mathrm{mod}}}(xI)^nT$, with the output qubit at site $N-2$.
The resulting ordered list is denoted $(c_0,\ldots,c_L)$.
This finite procedure also defines, for every physical bond,
\begin{equation}
W_b=\{((c_t)_b,(c_t)_{b+1}):0\le t\le L\}.
\label{eq:clock-pair-predicate}
\end{equation}
For a pair of valid basis states, membership in $W_b$ refers to the
underlying labels and permits every internal bit value.  Let $F_b$ be the
set of valid basis pairs whose label pair is absent from $W_b$.
Thus $F_b$ is an explicit whole-fibre clock predicate, computed in time
polynomial in the circuit size.

\paragraph*{Diagonal terms and the Hamiltonian.}
Let $I_b^{\mathrm{bad}}$ and $O_b^{\mathrm{bad}}$ be sets of valid basis
pairs.  The input sets are empty except for
\[
I_0^{\mathrm{bad}}=\{(A_1,o)\}\quad\text{if }0\in\mathcal A,
\qquad
I_{2j-1}^{\mathrm{bad}}=\{(A_0,I_1),(A_1,I_1)\}
\quad\text{if }j\in\mathcal A\setminus\{0\}.
\]
The only output set is $O_{N-2}^{\mathrm{bad}}=\{(I_0,T)\}$.
For $p(b+1)$ the parity of site $b+1$, set
\begin{equation}
D_b=\Id_{32}\otimes J_{p(b+1)}+
\sum_{z\in F_b\,\cup\,I_b^{\mathrm{bad}}\,\cup\,O_b^{\mathrm{bad}}}
\ket z\!\bra z.
\label{eq:raw-clock-diagonal}
\end{equation}
The valid pairs in this formula mean their encoded ten-bit words, and the
union counts each pair once.  The invalid sector is orthogonal to them.
At the boundaries define
\[
D_L=J_E+\sum_{z\in\{M_0,M_1,I_0,I_1,O\}}\ket z\!\bra z,
\qquad
D_R=\sum_{z\in\mathcal B_E\setminus\{O,T\}}\ket z\!\bra z.
\]
Here the set difference uses the states of the ordered basis.  Thus the
valid left boundary allows precisely $A_0,A_1,x,B,T,X$, and the right
boundary allows $O,T$.  Site $0$ has its invalid-code check in $D_L$;
every site $j\ge1$ has its check in $D_{j-1}$.  Each invalid site is
therefore assigned once.

Write $\mathcal T_b$ for the rows selected at bond $b$ and $P_{b,r}$ for
their projectors.  The complete raw Hamiltonian is
\begin{equation}
H_{\mathrm{rec}}=D_L+P_{\mathrm{EP}}+D_R+
\sum_{b=0}^{N-2}\left(D_b+\sum_{r\in\mathcal T_b}P_{b,r}\right).
\label{eq:raw-clock-hamiltonian}
\end{equation}
Each displayed summand is one raw projector.  The groups used for merging
are $\{D_L,P_{\mathrm{EP}}\}$, each physical-bond parenthesis, and
$\{D_R\}$.  There are five transition projectors on an even bond.
An odd bond has four base rules $2b,3b,4b,5a$, one end rule and one
activation rule, with $2a$ added exactly on $\mathcal B_{\mathrm{gate}}$.

The support decomposition of the merged groups is given in
Appendix~\ref{app:source-reflections}.

\subsection{Instance size and projector counts}
\label{app:source-normalization-count}

\begin{lemma}[Size of the one-dimensional history-state instance]
\label{lem:source-normalization}
After adjoining at most one idle ancilla wire initialized to $\ket0$ and, if necessary, one
complete identity sweep immediately before the final sweep, the nearest-neighbor verifier has $n,K\ge2$,
where $n$ is the total number of witness, ancilla, and output wires and $K$ is the number of
nearest-neighbor gate rounds.  This padding preserves the verifier's acceptance probability
exactly.  The one-dimensional history-state construction has
\[
N=2nK+1,\qquad L=(K-1)n(4n+6)+2n,
\qquad M_{\mathrm{ung}}=(14n-1)K+3.
\]
Here $N$ is the number of $d=11$ source sites, $L$ is the history length, and
$M_{\mathrm{ung}}$ counts the local projector occurrences before the bondwise merge.
\end{lemma}
\begin{proposition}[Gap of the unmerged history-state Hamiltonian]
\label{prop:ambient-raw-gap}
Apply the unmerged history-state projector presentation of Theorem~\ref{thm:rgn}, with the reconstruction of
Lemma~\ref{lem:source-normalization}, to a normalized NO circuit from
Lemma~\ref{lem:ambient-normalization}.  For the final parameters $n,K$ of
Lemma~\ref{lem:source-normalization}, define
\[
N=2nK+1,\qquad L=(K-1)n(4n+6)+2n,\qquad m_\star=16(N+1)^3.
\]
Then the raw source Hamiltonian $H_{\mathrm{rec}}$, including every parity-encoded local projector
occurrence before the bondwise merge, satisfies
\[
H_{\mathrm{rec}}\succeq\frac1{G_{\mathrm{gap}}}\Id,\qquad
G_{\mathrm{gap}}=1+28Nm_\star+2m_\star^2+4(L+1)^3.
\]
On a YES circuit its kernel is the accepting history space.
\end{proposition}

\begin{proof}[Proof of Lemma~\ref{lem:source-normalization}]
The generator in Appendix~\ref{app:clock-constructor} has $(K-1)n$ complete modules
followed by a final forward sweep ending at site $2nK$, so $N=2nK+1$.
Each complete module contributes $2n$ forward transitions, two turning transitions,
$2n+1$ return transitions, and three handoff transitions.  Therefore
\[
L=(K-1)n\bigl[2n+2+(2n+1)+3\bigr]+2n
=(K-1)n(4n+6)+2n.
\]

It remains to count the static projectors.  For a group $b$, let $m_b$ be its number
of raw projector occurrences.  Each of the $K$ rounds has $n-1$ odd gate bonds and one odd
nongate bond, giving $K(n-1)$ gate groups and $K$ nongate groups.  The counts by support
type are
\[
\begin{array}{l|c|c}
\text{support type}&\text{number of groups}&m_b\\ \hline
\text{left boundary}&1&2\\
\text{even physical bond}&nK&6\\
\text{odd gate bond}&K(n-1)&8\\
\text{odd nongate bond}&K&7\\
\text{right boundary}&1&1
\end{array}
\]
Thus the largest group multiplicity is eight and
\[
M_{\mathrm{ung}}=2+6nK+8K(n-1)+7K+1=(14n-1)K+3.
\]
The number of groups is $1+nK+K(n-1)+K+1=2nK+2=N+1$.
\end{proof}

\subsection{Transition paths and spectral gap}

\paragraph*{Safe-label decomposition and transition graph.}
Fix the normalized source construction of Proposition~\ref{prop:ambient-raw-gap}, with
parameters $n,K,N,L,m_\star$ as defined there.
Decode the five-bit label registers.  The source Hilbert space decomposes as
\[
\mathcal H_{\mathrm{src}}=\mathcal H_{\mathrm{inv}}\oplus
\bigoplus_{w\in\mathcal W_{\mathrm{dec}}}\ket w\otimes\mathcal F_w,
\]
where $\mathcal H_{\mathrm{inv}}$ contains configurations with an invalid codeword,
$\mathcal W_{\mathrm{dec}}$ is the set of parity-valid decoded label words, and $\mathcal F_w$
is the tensor product of the internal qubits carried by the $A$, $M$, and $I$ labels.  Let
$\mathcal W_{\mathrm{safe}}\subseteq\mathcal W_{\mathrm{dec}}$ contain the words that pass every
whole-fibre clock and boundary check, where a whole-fibre diagonal term at $w$ acts as
$\ket w\!\bra w\otimes\Id_{\mathcal F_w}$, and set
\[
\mathcal S=\bigoplus_{w\in\mathcal W_{\mathrm{safe}}}\ket w\otimes\mathcal F_w.
\]
We call $\mathcal S^\perp$ the rejected-label sector.  Let $D_{\mathrm{diag}}$ be the diagonal
part of $H_{\mathrm{rec}}$.
The input and output diagonal projectors are not used to define $\mathcal W_{\mathrm{safe}}$
because they test proper subspaces of $\mathcal F_w$ on the canonical label path.  Retaining the
invalid-code projectors and the whole-fibre clock and boundary projectors gives a positive
diagonal suboperator $0\preceq D_{\mathrm{full}}\preceq D_{\mathrm{diag}}$ with
\[
D_{\mathrm{full}}\succeq\Pi_{\mathcal S^\perp}.
\]
For each fixed assignment outside its support, a transition projector decomposes into mutually
orthogonal blocks on $\mathcal F_w\oplus\mathcal F_{w'}$ of the form
\[
\frac12
\begin{pmatrix}
\Id_{\mathcal F_w}&-V^\dagger\\
-V&\Id_{\mathcal F_{w'}}
\end{pmatrix},
\]
where $V:\mathcal F_w\to\mathcal F_{w'}$ is unitary.  Let $\Gamma_{\mathrm{safe}}$ be the graph
whose vertices are $\mathcal W_{\mathrm{safe}}$ and whose edges are the transition blocks with
both endpoints in $\mathcal W_{\mathrm{safe}}$.

\begin{lemma}[Clock-word classification and transition paths]
\label{lem:clock-word-classification}
For the normalized source construction and safe transition graph defined above,
every safe source word contains exactly one active label, and the number of safe words is at most
$16N(N+1)^2\le m_\star$. The safe transition graph has one canonical boundary-to-boundary
component, which carries $n$ qubit labels and reproduces the verifier history. Every other
component is a path with at most $m_\star$ vertices and has a transition from an endpoint to the
rejected-label sector.
\end{lemma}
\begin{proof}
Orient each transition by its arrow in Table~\ref{tab:raw-clock-rules}, and
orient the boundary move as $B\to T$.  We first extract the local conditions
needed from the canonical generator; a word satisfying these conditions
need not be safe.

Partition the labels into done, active, and ready classes by site parity:
\[
\begin{array}{c|ccc}
 &\mathsf D&\mathsf A&\mathsf R\\ \hline
E&\{X,x\}&\{A,M,T,B\}&\{I,O\}\\
O&\{X,I\}&\{A,M,T,B\}&\{O,o\}.
\end{array}
\]
Every pair in $W_b$ has class $\mathsf{DD}$, $\mathsf{DA}$,
$\mathsf{AR}$, or $\mathsf{RR}$.  The possible same-class pairs are
$XX,xI,Ix,Xx$ in the done region and $Io,oI,oO,OO$ in the ready region,
with parity fixed by the alphabets.  The possible neighbors of an active
label are as follows; $\partial$ denotes the corresponding physical boundary.
\[
\begin{array}{c|c|c|c}
\text{active label}&\text{site parity}&\text{left neighbor}&\text{right neighbor}\\ \hline
A,M&E&I,X,\partial&o\\
A,M&O&x&I,O\\
B&E&I,X,\partial&o\\
B&O&x&I,O\\
T&E&I,X,\partial&O,o,\partial\\
T&O&x,X&I,O
\end{array}
\]
At the left boundary an active $A$ is allowed but $M$ is not.
The position-dependent conditions needed below are
\[
\begin{array}{c|c}
\text{pair in }W_b&\text{required bond set}\\ \hline
AI,IA&\mathcal B_{\mathrm{gate}}\\
AO&\mathcal B_{\mathrm{end}}\\
MO&\mathcal O\setminus\mathcal B_{\mathrm{end}}\\
XA&\mathcal B_{\mathrm{act}}\\
XM&\mathcal O\setminus\mathcal B_{\mathrm{act}}
\end{array}
\]
These necessary conditions follow directly from the generator in
Appendix~\ref{app:clock-constructor}.  Sweep $s$ starts with
$X^{2s}H_s o(Io)^{n-1}O^{N-2s-2n}$, where $H_s=A$ precisely when
$s$ is a multiple of $n$.  Its forward moves transfer the alternating
ready segment into the done segment, and the return moves carry $B$
back through that segment before the handoff adds two $X$ labels.
The displayed passive pairs and active neighbors are preserved through
these moves, including the turn and boundary move.  An $A$ sweep with
$s=rn$ produces $AI,IA$ only on bonds $2rn+2j+1$, $j<n-1$, and $AO$
only on its end bond.  A messenger sweep ends on the complementary
bonds.  The pairs $XA,XM$ occur at activation into the corresponding
next sweep, giving the last two rows.

The boundary checks exclude ready labels on the left and done labels on
the right.  The class-pair restriction therefore forces exactly one
active label.  The passive pairs force its prefix and suffix to be
\[
X^\alpha(xI)^u x^\varepsilon,
\qquad o^\eta(Io)^vO^\beta,
\qquad \varepsilon,\eta\in\{0,1\}.
\]
For a fixed active position, $(\alpha,\varepsilon)$ determines $u$ and
$(\beta,\eta)$ determines $v$.  Hence there are at most
$4N[2(N+1)]^2=16N(N+1)^2\le m_\star$ safe words.

\emph{One incoming and one outgoing transition.}
Every transition touches an active label, so only its two adjacent pairs
can participate.  For an even $A$ or $M$, the left neighbor $I$ or $X$
selects an incoming rule $2$ or $6$, and rule $1$ is outgoing.
For an odd $A$ or $M$, rule $1$ is incoming and the right neighbor selects
the appropriate gate, messenger, or end rule.  The position conditions
above guarantee that the required rule is present.
For an even $B$, rule $4a$ is incoming, while the left neighbor selects
$4b$, $5a$, or the boundary move; for an odd $B$, the right neighbor
selects incoming $4b$ or $3b$, and $4a$ is outgoing.
For an even $T$, its left neighbor selects the incoming end rule, $5a$,
or the boundary move, and its right neighbor selects outgoing $3a$ or
$5b$.  For an odd $T$, the left neighbor selects incoming $3a$ or $5b$,
and the right neighbor selects outgoing $3b$ or the activation rule.
Here rule numbers without suffixes denote the matching row in
Table~\ref{tab:raw-clock-rules}.  The only absent incoming edge is at
an $A$ on site $0$, and the only absent outgoing edge is at a $T$ on
site $N-1$.  These statements concern the full transition graph:
either neighbor may be rejected by a clock or boundary check.

\emph{A strictly increasing clock coordinate.}
Let $p$ be the active position and let $d(w)$ count the $X$ labels.
Define
\[
\begin{split}
r(w)&=(2N+1)d(w)+u(w),\\
u(w)&=\begin{cases}
p,&\text{head }A,M,\text{ or head }T\text{ followed by }O\text{ or }\partial,\\
2N-1-p,&\text{head }B,\\
2N,&\text{head }T\text{ followed by }o\text{ or }I.
\end{cases}
\end{split}
\]
The neighbor table covers every $T$ context, so $u$ is defined on every
safe word and $0\le u\le2N$.  Each forward edge with two safe endpoints
strictly increases $r$.  Indeed, rules $1$, $2$, and $3a$ move a
right-moving head one site right, increasing $u$ (or raising it to
$2N$).  Rules $4a,4b$ move $B$ one site left and increase $u$ by one.
Rule $3b$ increases $u$ by $2N-1-2p>0$, since $p\le N-2$.
For $5a$ and the boundary move, the even $B$ has right neighbor $o$,
so $u$ rises from $2N-1-p$ to $2N$.
These rules preserve $d(w)$.  Finally, $5b$ and $6$ increase $d(w)$
by one, which outweighs any decrease of $u$.

The safe graph has degree at most two.  A cycle would have a vertex
maximizing $r$ with two incoming edges, contradicting the local
uniqueness just proved.  Each component is therefore a finite path.
At a path endpoint, a missing safe neighbor is an escape whenever the
corresponding full-graph edge exists.  Thus a component with no escape
must start with $A$ at site $0$ and end with $T$ at site $N-1$.

Let $q$ be the conserved number of qubit-bearing labels on such a
component.  Its initial word is $Ao(Io)^{q-1}O^{N-2q}$, so $q\ge1$.
If $q<n$, its first sweep produces $AO$ at bond $2q-1$, which is not
an end bond.  If $q>n$, it produces $AI$ at bond $2n-1$, which is not
a gate bond.  Either gives an escape, so $q=n$.
For $q=n$, the initial word is $c_0$ and the unique outgoing transitions
are exactly those of the canonical generator.  Every canonical word
is safe by the definition of $W_b$ and the boundary checks; the local
incoming/outgoing uniqueness excludes additional edges from that path.
The $A$ sweeps apply the scheduled verifier gates and the other moves
transport the internal qubits.  Hence this is the unique canonical
component, and all other components have endpoint escapes.
The word count bounds their sizes.  Finally, clock and boundary
rejection acts on whole internal fibres, and each transition is unitary
between its endpoint fibres, giving the asserted full-fibre escape.
\end{proof}

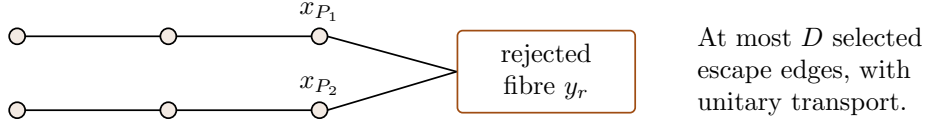
\begin{figure}[htbp]
\centering
\begin{tikzpicture}[x=1cm,y=.85cm,font=\small]
  \node[anchor=west,font=\small\bfseries] at (0,0) {Canonical history component};
  \foreach \x in {1,3,5,7,9,11,13} \draw[fill=diagramblue!12] (\x,-.7) circle (3pt);
  \foreach \x in {1,3,5,7,9,11} \draw (\x+.1,-.7)--(\x+1.9,-.7);
  \node[anchor=west] at (0,-1.4) {Each vertex is a whole-chain label word with an internal qubit fibre.};
  \node[anchor=west,font=\small\bfseries] at (0,-2.25) {Noncanonical safe paths can meet the same rejected fibre};
  \foreach \x in {1,3,5} \draw[fill=diagramorange!12] (\x,-3.1) circle (3pt);
  \draw (1.1,-3.1)--(2.9,-3.1); \draw (3.1,-3.1)--(4.9,-3.1);
  \foreach \x in {1,3,5} \draw[fill=diagramorange!12] (\x,-4.25) circle (3pt);
  \draw (1.1,-4.25)--(2.9,-4.25); \draw (3.1,-4.25)--(4.9,-4.25);
  \node[anchor=south] at (5,-3.02) {$x_{P_1}$};
  \node[anchor=south] at (5,-4.17) {$x_{P_2}$};
  \node[diagram box,draw=diagramorange,text width=2cm] (r) at (8,-3.65) {rejected\\fibre $y_r$};
  \draw (5.1,-3.1)--(r.west); \draw (5.1,-4.25)--(r.west);
  \node[align=left,text width=4cm] at (12,-3.65)
    {At most $D$ selected\\escape edges, with\\unitary transport.};
\end{tikzpicture}
\caption{The configuration graph used in Lemma~\ref{lem:clock-word-classification}.
Vertices are whole-chain label configurations, not physical sites; edges are reversible
transitions.  The canonical component follows the scheduled computation.
Every noncanonical safe path has an endpoint escape to a rejected fibre.
The labels $x_{P_1},x_{P_2}$ and $y_r$ denote amplitudes.
At most $D$ selected escape edges meet each rejected fibre; its penalty and these transitions
jointly control the safe endpoints.  Shared rejected fibres can couple different safe paths,
so the gap proof treats the full noncanonical complement.}
\label{fig:clock-configuration-graph}
\end{figure}

\begin{proof}[Proof of Proposition~\ref{prop:ambient-raw-gap}]
Use the decomposition and path classification of Lemma~\ref{lem:clock-word-classification},
illustrated in Figure~\ref{fig:clock-configuration-graph}.
\par\smallskip\noindent\emph{The noncanonical-sector gap.}
Let the noncanonical block be the orthogonal sum of all reducing components other than the
canonical history path.  Let $\psi$ lie in this block and write
$E=\langle\psi,H_{\mathrm{rec}}\psi\rangle$.  Choose one such boundary transition at an endpoint of
each noncanonical path.  The source Hamiltonian contains
\[
T_{\mathrm{tr}}=5nK+6nK+K(n-1)+1=6N-K-5<7N
\]
transition projectors.  Each projector connects a fixed label fibre to at most one other fibre,
so a rejected label fibre is incident to at most $D=6N-K-5$ chosen boundary transitions.
For each path $P$, write $x_P$ and $y_{r(P)}$ for the chosen safe and rejected endpoint amplitudes,
and let $U_P$ be the unitary transport from the rejected fibre to the safe fibre.
The diagonal penalty and the orthogonal selected blocks within each raw projector occurrence give
\[
E\ge\lVert\Pi_{\mathcal S^\perp}\psi\rVert^2
+\frac12\sum_P\lVert x_P-U_Py_{r(P)}\rVert^2.
\]
Each rejected fibre occurs at most $D$ times in this sum.  Since $2\le D<7N$,
the squared triangle inequality gives
\[
\begin{aligned}
\sum_P\lVert x_P\rVert^2
&\le2\sum_P\lVert x_P-U_Py_{r(P)}\rVert^2
+2D\lVert\Pi_{\mathcal S^\perp}\psi\rVert^2\\
&\le2DE\le14NE.
\end{aligned}
\]

Gauge away the unitary transports on a component of $\Gamma_{\mathrm{safe}}$ with amplitudes
$z_0,\ldots,z_{m-1}$ and differences $d_j=z_{j+1}-z_j$. Telescoping and
Cauchy--Schwarz give
\[
\sum_j\lVert z_j\rVert^2
\le2m\lVert z_0\rVert^2+m^2\sum_j\lVert d_j\rVert^2.
\]
Each propagation block contributes $\frac12\lVert d_j\rVert^2$. Summing over the disjoint
components and adding the $\mathcal S^\perp$ norm controlled by $D_{\mathrm{full}}$ yields
\[
\lVert\psi\rVert^2\le
(1+28Nm_\star+2m_\star^2)
\langle\psi,H_{\mathrm{rec}}\psi\rangle
\]
on the noncanonical reducing block. Its least eigenvalue is at least
$G_{\mathrm{ill}}^{-1}$ for
\[
G_{\mathrm{ill}}=1+28Nm_\star+2m_\star^2.
\]

\par\smallskip\noindent\emph{The canonical history-path gap.}
On the canonical block, conjugate propagation by the ordered verifier circuit encoded by the
canonical path. The result is one half
of the path-graph Laplacian on $L+1$ vertices, whose first nonzero eigenvalue obeys
\[
1-\cos\frac\pi{L+1}\ge\frac2{(L+1)^2}.
\]
The temporally unique input projector terms and the final output projector term have orthogonal
clock supports.  Let $U$ denote the verifier unitary obtained by composing the gates along the
canonical path, and let $P_{\mathrm{reject}}$ denote the rejecting-output projector on the
verifier registers.  After pulling the input and output projectors back to time zero, the input
sum dominates the projector $P_{\mathrm{badAnc}}$ onto improperly initialized ancillas.  With
$P_0=1-P_{\mathrm{badAnc}}$ and $R=U^\dagger P_{\mathrm{reject}}U$, constant soundness gives
\[
P_0RP_0\succeq\frac12P_0.
\]
The two-projector angle formula then gives
$P_{\mathrm{badAnc}}+R\succeq(1-1/\sqrt2)\Id\succeq\frac14\Id$. A normalized vector in
the propagation kernel consequently has diagonal expectation at least $1/[4(L+1)]$. 
Kitaev's geometric lemma~\cite{Kitaev2002} for the propagation and diagonal kernels, together with
$2\sin^2(\theta/2)\ge\sin^2(\theta)/2$, gives
\[
H_{\mathrm{rec}}|_{\mathrm{can}}
\succeq\frac1{4(L+1)^3}\Id
=\frac1{G_{\mathrm{can}}}\Id,\qquad
G_{\mathrm{can}}=4(L+1)^3.
\]

\par\smallskip\noindent\emph{Combination of the reducing sectors.}
The canonical path and its orthogonal complement are reducing subspaces. Hence
\[
H_{\mathrm{rec}}\succeq
\min\{G_{\mathrm{ill}}^{-1},G_{\mathrm{can}}^{-1}\}\Id
\succeq\frac1{G_{\mathrm{ill}}+G_{\mathrm{can}}}\Id,
\]
which is the claimed value of $G_{\mathrm{gap}}$. On a YES instance, the same canonical
history is annihilated by propagation, initialization, and output terms, giving the stated
kernel.
\end{proof}

\section{Sparse realizations of Hermitian and nilpotent supercharges}

\subsection{Listed Hermitian supercharges}
\label{app:n1-hermitian-relation}

\begin{lemma}[Sparse relation for a listed Hermitian supercharge]
\label{lem:n1-hermitian-sparse-relation}
Let $x$ be a well-formed promised input of
Definitions~\ref{def:1glsp-presentation}--\ref{def:1glsp}, of physical length $S$.  There is a
uniformly specified exact sparse relation $R_x$ such that
\[
\ker R_x\cong\ker Q_x
\]
with equal nullity.  On every NO input,
\[
\sigma_{\min}(R_x)\ge\frac1{P_{\rm rel}(S)},
\]
where
\[
P_{\rm rel}(S)=
4000S^5+8000S^4+4220S^3+220S^2+201S+200.
\]
The construction preserves the order and multiplicity of the local terms and permits exact
cancellation between distinct encoded entries and terms.
\end{lemma}
\begin{proof}
Enumerate the encoded nonzero local matrix entries as
$e=0,\ldots,N_{\rm rec}-1$.  Since every entry occupies at least one physical input bit,
$N_{\rm rec}\le S$.  An entry $e$ belongs to an ordered term $Q_a$, has declared support
$J_a=(j_1,\ldots,j_r)$, declared local basis order, local row and column words $u_e,v_e$, and
coefficient
\[
\alpha_e=\frac{a_e+b_e\sqrt2+i(c_e+d_e\sqrt2)}{2^{t_e}}.
\]

Let $J_a^\uparrow$ be the increasing rearrangement of the support and apply the same coordinate
permutation to $u_e,v_e$.  The local basis words are restrictions of the fixed site-major
occupation basis, not declared-order wedge vectors, so this normalization introduces no extra
phase in $\alpha_e$.  For a global occupation word $z$, entry $e$ has an occurrence precisely
when $z|_{J_a^\uparrow}=v_e^\uparrow$.  Apply the graded matrix unit
\[
E_{J_a^\uparrow}(u_e^\uparrow,v_e^\uparrow)=
\left(\prod_{q=1}^{r}(c_{j_q^\uparrow}^\dagger)^{(u_e^\uparrow)_q}\right)
\left(\prod_{q=1}^{r}(1-n_{j_q^\uparrow})\right)
\left(\prod_{q=r}^{1}(c_{j_q^\uparrow})^{(v_e^\uparrow)_q}\right).
\]
Sequential CAR action either vanishes or produces a unique target occupation word $\tau(e,z)$
and a sign $\chi(e,z)\in\{+1,-1\}$.  The sign is computed in the global site-major order and
remains separate from the input coefficient.

Write
\[
\mathcal O=\{(e,z):z|_{J_a^\uparrow}=v_e^\uparrow\}
\]
for the occurrence set.  For a fixed source or target word there are at most $N_{\rm rec}$
occurrences; put $D=\max(1,N_{\rm rec})\le S+1$.  Entries belonging to repeated terms or repeated
local tables retain different ordinals.  No multiplicity is removed before target aggregation.

Because $\lVert Q_a\rVert\le1$, every local matrix entry satisfies $|\alpha_e|\le1$.  The signed
coordinates and the binary representation of $t_e$ occupy at most $S$ bits.  Apply
Lemma~\ref{lem:stable-quadratic-dyadic-multiplier} to every occurrence $o=(e,z)$.  It introduces
a unique auxiliary history $h_o$, residual $g_o$, and endpoint $y_o$ such that zero residual gives
$y_o=\alpha_e\psi_z$.  With
\[
P_{\rm end}(S)=20(S+1),
\qquad
P_{\rm hist}(S)=100(S+1),
\]
the shared estimates are
\[
|y_o-\alpha_e\psi_z|
\le P_{\rm end}(S)\lVert g_o\rVert,
\qquad
\lVert h_o\rVert
\le P_{\rm hist}(S)\bigl(|\psi_z|+\lVert g_o\rVert\bigr).
\]
The domain of $R_x$ consists of the core vector $\psi$ and all multiplier histories.  Its codomain
consists of the multiplier residuals and one aggregation row for each global target word $i$.
Define
\[
a_i=\sum_{\substack{o\in\mathcal O\\\tau(o)=i}}\chi(o)y_o,
\qquad
R_x(\psi,h)=(g,a).
\]
If $R_x(\psi,h)=0$, every multiplier history is exact and $a_i=(Q_x\psi)_i$.  Thus
$\psi\in\ker Q_x$.  Uniqueness of the histories makes projection onto $\psi$ injective on
$\ker R_x$.  Conversely, each $\psi\in\ker Q_x$ has one exact history collection and makes every
aggregation row vanish.  This proves the kernel and nullity identity, including cancellation
between different encoded entries and terms.

For the NO bound, write $r=(g,a)=R_x(\psi,h)$.  The endpoint-to-target aggregation map has one
nonzero per occurrence column and at most $D$ per target row, hence norm at most $\sqrt D$.  Thus
\[
\lVert Q_x\psi\rVert
\le\lVert a\rVert+\sqrt D\,P_{\rm end}(S)\lVert g\rVert
\le C(S)\lVert r\rVert,
\qquad
C(S)=1+20S(S+1).
\]
The NO promise gives $Q_x^2\succeq G_I^{-1}\Id$ and $G_I\le S$, so
\[
\lVert\psi\rVert
\le\sqrt{G_I}\,\lVert Q_x\psi\rVert
\le S C(S)\lVert r\rVert.
\]
Each source coordinate feeds at most $S$ histories.  Summing the history estimate in squared norm
gives
\[
\lVert h\rVert
\le2P_{\rm hist}(S)\bigl(S^2C(S)+1\bigr)\lVert r\rVert.
\]
Consequently
\[
\lVert(\psi,h)\rVert
\le\Bigl(SC(S)+2P_{\rm hist}(S)(S^2C(S)+1)\Bigr)\lVert R_x(\psi,h)\rVert,
\]
and the polynomial in parentheses is exactly $P_{\rm rel}(S)$.  If $N_{\rm rec}=0$, then
$Q_x=0$, the input cannot satisfy the NO promise, and the relation has the full core space as its
kernel, as required.
\end{proof}

\begin{proposition}[Admissible Hermitian realization of the sparse relation]
\label{prop:n1-hermitian-realization}
For every promised listed Hermitian input $x$ of length $S$, a deterministic polynomial-time
algorithm specifies a Hermitian exact sparse matrix $A_x$ on
\[
m_A=32(S+1)
\]
qubits such that $\ker A_x\cong\ker Q_x$ with equal nullity.  On every NO input,
\[
\sigma_{\min}(A_x)
\ge\frac1{P_{\rm rel}(S)^2B_{\rm N1}(S)^2},
\qquad
B_{\rm N1}(S)=40(S+1)^2+2.
\]
The matrix has norm at most one, coefficients in $K_8$, polynomial row and column sparsity, and
independent clean exact row, column, value, and adjoint access.  It has one positive common
denominator and coordinate height bounded numerically by a polynomial in $S$.
\end{proposition}
\begin{proof}
The deterministic total decoder validates the site-major allocation, ordered term list, support
and basis permutations, strictly ordered nonzero table entries, independent binary exponents,
unary gap word, and exact end of input. It accepts only the unique canonical representation.
A malformed encoding selects one fixed identity matrix.  For each nonzero coefficient the
decoder applies the multiplier guard of Appendix~\ref{app:stable-quadratic-dyadic-multiplier};
if any guard rejects, it selects the same fixed identity matrix.  Every promised input passes
these guards because each supplied term has norm at most one.
For an otherwise valid construction, every invalid logical label is assigned its own diagonal
identity entry.  A valid column label carries a type field distinguishing core-vector from
multiplier-history coordinates, and a valid row label carries a type field distinguishing
multiplier residuals from target-aggregation rows.  With $u=S+1$, recomputing all redundant
fields and requiring unused fields to vanish gives a canonical injection into $32u$ bits.

Let $R_x$ be the relation in Lemma~\ref{lem:n1-hermitian-sparse-relation}. Its row degree is at
most $\max(D,3)$; a core column enters at most $D(3S+5)$ multiplier rows, and every other
column has constant degree. Thus both sparsities are at most $10(S+1)^2$, and the incidence
bound gives
\[
\lVert R_x\rVert\le L_R(S):=20(S+1)^2.
\]
For each query, the access circuit scans the at most $N_{\rm rec}\le S$ encoded entries,
reconstructs compatible source or target labels, recomputes the CAR sign, and aggregates equal
outputs exactly. Addressed multiplier stages avoid a loop of numerical length $t_e$.
The output entries are copied before all decoding, selection, arithmetic, and sorting work is
reversed, giving independent clean row, column, value, and adjoint-value access.

Apply Theorem~\ref{thm:common-kernel-exact-sparse-transfer} to the single residual
$B_1=R_x$, with $\ell=L_R(S)$ and
$\delta_x=P_{\rm rel}(S)^{-2}$. Its integer normalization is
\[
\beta=2L_R(S)+2=B_{\rm N1}(S).
\]
The resulting matrix $A_x$, including its invalid-label identity blocks, has the stated
kernel, nullity, and NO singular-gap bound.
Before normalization the finite coefficient alphabet has denominator at most two; afterward a
common denominator and every integral-basis coordinate are bounded by
$2B_{\rm N1}(S)$. These numerical bounds, the $32u$ label width, and the access resources are
all polynomial in $S$, proving the proposition.
\end{proof}

\subsection{Explicit nilpotent-supercharge reduction}
\label{app:explicit-n2-reduction}

In the ordered basis $00,01,10,11$, set
\[
H_L:=\widehat H\otimes I_2,\quad H_R:=I_2\otimes\widehat H,\quad
T_L:=T\otimes I_2,\quad T_R:=I_2\otimes T,
\]
and
\[
\mathrm{CX}_{L\to R}
:=\sum_{a,b\in\{0,1\}}\ket{a,a\oplus b}\!\bra{a,b},\qquad
\mathrm{CX}_{R\to L}
:=\sum_{a,b\in\{0,1\}}\ket{a\oplus b,b}\!\bra{a,b}.
\]

Begin with the normalized verifier supplied by Lemma~\ref{lem:ambient-normalization}.
For the source parameters $n,K$ of
Lemma~\ref{lem:source-normalization}, Proposition~\ref{prop:ferm} gives
\[
N=2nK+1,\qquad M_0=N+1=2nK+2.
\]
Each bondwise-merged even projector $\Pi_c$ acts on at most ten consecutive encoded modes.  Assign its
anchored label Majorana $\Gamma_c$, put $q_c=\Pi_c\Gamma_c$, and set
$Q_{\min}=\sum_cq_c$.  Let $G_{\mathrm{gap}}$ be the explicit integer from
Proposition~\ref{prop:ambient-raw-gap}.  Lemma~\ref{lem:source-grouping} and
Theorem~\ref{thm:minimal-hardness} give, on NO instances,
\[
Q_{\min}^2\succeq
\frac{1}{64M_0G_{\mathrm{gap}}^2}\Id.
\]

Lemma~\ref{lem:source-routing} gives an exact gate sequence over $\mathcal G$ for each scheduled
two-site operation.  The resulting matrix belongs to the scheduled family
\[
I_4,\quad H_L,\quad H_R,\quad T_L,\quad T_R,
\quad \mathrm{CX}_{L\to R},\quad \mathrm{CX}_{R\to L}.
\]
For a two-qubit gate $U$, Eq.~\eqref{eq:raw-clock-projector} gives, in the ordered
domain/image basis,
\[
P_U=\frac12
\begin{pmatrix}
\Id_4&-U^\dagger\\
-U&\Id_4
\end{pmatrix},
\qquad
\Id_8-2P_U=
\begin{pmatrix}
0&U^\dagger\\
U&0
\end{pmatrix}.
\]
Order the source projector groups of Appendix~\ref{app:source-normalization-count} from the
left boundary through the bonds $0,1,\ldots,N-2$ to the right boundary.
Distinct transition blocks occupy disjoint ordered
label-pair sectors.  The support-block decomposition in
Appendix~\ref{app:source-reflections} intersects these blocks with the diagonal predicates and yields
only graph projectors, complete endpoint identity blocks, channelwise partial identity blocks,
residual diagonal identities, or the gate-input partial block
\begin{equation}
\label{eq:gate-input-partial-block}
\Id_8-\frac12
\begin{pmatrix}
Q_{\mathrm{good}}&Q_{\mathrm{good}}U^\dagger\\
UQ_{\mathrm{good}}&UQ_{\mathrm{good}}U^\dagger
\end{pmatrix},
\end{equation}
where $Q_{\mathrm{good}}=I_p\otimes\ket0\!\bra0_q$ is the rank-two diagonal projector onto the
$q=0$ input channel.  In the endpoint packing $(e,p,q)$, the bit $e$ distinguishes the domain
and image, while $(p,q)$ index their four internal states.  Orthogonal direct sums of these cases cover every local term
arising from the scheduled family.

Encode each $d=11$ source site in five even-weight modes.  Anchor one label mode for each
group at the leftmost coordinate of its projector window, with labels ordered after the
system mode at a shared coordinate.  Add two fresh terminal modes $a<b$.  The mode count is
\[
m=5N+M_0+2=6M_0-3.
\]
For each $q_c=\Pi_c\Gamma_c$, form $\mathcal Q_c=q_cs$ using
Lemma~\ref{lem:nilpotent-two-mode-extension}.  Its increasing support contains at most ten
system modes, one
label mode, and $a,b$.  Evaluate the complete local table in the declared occupation order by
applying operators right to left with the CAR prefix signs, and aggregate every entry exactly in
$\mathbb Q(\zeta_8)$.

Every matrix in the scheduled family has a representation with common denominator $2$ and
integral-basis coordinate height at most $2$, using
\[
\widehat H=\frac1{\sqrt2}
\begin{pmatrix}1&1\\1&-1\end{pmatrix},
\qquad
\frac1{\sqrt2}=\frac{\zeta_8-\zeta_8^3}{2}.
\]
The merge cases above use only diagonal identities, graph blocks, and the displayed gate-input
block.  For an entry of $UQ_{\mathrm{good}}U^\dagger$, at most four intermediate indices contribute, while
multiplication in $\mathbb Z[\zeta_8]/(\zeta_8^4+1)$ is a four-term signed convolution.
Bounding these contributions and the identity entry at common denominator $8$ gives
\[
\frac18\sum_{j=0}^3z_j\zeta_8^j,
\qquad |z_j|\le72.
\]
Multiplication by $\Gamma_c$ and $s$ only permutes entries and changes signs.  Taking the gcd of
$8$ and all table coordinates and dividing globally yields a positive reduced denominator
$h\le8$ without increasing the height.

With
\[
v=64M_0G_{\mathrm{gap}}^2,
\qquad u=1,
\qquad Z=m+v+128,
\]
append exactly $Z$ ordinary one-mode zero terms, each given by its complete local zero table.
They change no operator, gap, or support.  The final list has $M=M_0+Z$ terms.  Each encoded term
contributes at least one bit to the list, so $S\ge M\ge Z$ and
\[
m,M\le S,
\quad h\le8\le S,
\quad |z_j|\le72\le S,
\quad u,v\le S,
\quad \delta=1/v\ge S^{-1}\ge S^{-2}.
\]
The construction is polynomial because the table arity is fixed and $M_0$, $m$, $v$, and $Z$
are polynomial in the source length.  Lemma~\ref{lem:nilpotent-two-mode-extension} transfers
the zero mode and the NO lower bound.  The canonical normal-ordered form of
$\Id_{ab}$ is the empty word, so the Hamiltonian keeps range $19$ and realized span at most
$15$.

\subsection{Exact reflections for the history-state projectors}
\label{app:source-reflections}

\paragraph*{Support decomposition of a merged group.}
The two endpoint label sectors of each selected row in
Table~\ref{tab:raw-clock-rules} are disjoint from those of every other
selected row on the same bond.  This follows by comparing the even rows
$1a,1b,3a,4a,5b$ and the odd rows separately; the two end rows and the two
activation rows are exclusive choices.  At the left boundary, both $B$
and $T$ lie in the kernel of $D_L$.

The clock part of $D_b$ either admits or rejects an entire endpoint
label sector because $W_b$ does not restrict internal bits.  Comparing
the input and output atoms with the table leaves exactly three possible
partial endpoint checks.  They are channel $1$ of the $1a$ domain at
$b=0$, the two-dimensional subspace $\operatorname{span}\{\ket{A_p,I_1}\}$
of the $2a$ domain at $b=2j-1$ for $j\in\mathcal A\setminus\{0\}$,
and channel $0$ of the $2A$ image at $b=N-2$.  Each of these endpoint
label pairs belongs to $W_b$ by the canonical procedure, as does its
opposite endpoint, so no whole-fibre clock penalty is present there.

Consequently, on an endpoint block $\mathcal D\oplus\mathcal I$, the
merged range projector is $P_V$ when neither endpoint is diagonal, and
is the identity when at least one complete endpoint is diagonal.
In the initial and final partial cases, the selected bad channel gives
the identity on its two endpoints and the other channel retains its
rank-one propagation projector.  In the gate/input case, identify the
domain with $\mathbb C^2\otimes\mathbb C^2$ and put
$Q_{\mathrm{good}}=\Id_2\otimes\ket0\!\bra0$.  The merged range projector has kernel consisting exactly
of $2^{-1/2}(v,Uv)$ with $v\in\operatorname{ran}Q_{\mathrm{good}}$, which
gives the rank-six projector displayed in
Eq.~\eqref{eq:gate-input-partial-block}.
All remaining forbidden valid atoms and the right-invalid sector are
orthogonal diagonal summands.  Every other basis sector, including the
left-invalid, right-valid sector, has zero local block.  These cases
exhaust the encoded bond space.

\begin{proof}[Proof of Lemma~\ref{lem:fieldtransfer}]
For a raw diagonal projector, phase each listed computational-basis atom.  An ordinary graph
projector is a direct sum of endpoint pairs, so its reflection is the corresponding product of
basis transpositions.  For a gate block, let $e$ distinguish the four-dimensional domain from
the four-dimensional image and put $D=C_e(U)=\operatorname{diag}(I_4,U)$.  In the packed basis,
\[
I-2P_U=D X_e D^\dagger.
\]
The exact scheduled word for $U$ supplies $D$ and $D^\dagger$.  Boundary projectors are one-pair
graph blocks, while input, output, clock, and invalid-code projectors are diagonal cases.

For a bondwise-merged group, the decomposition above gives the mutually orthogonal
partition of the ten-qubit bond basis into the right-invalid sector, residual valid diagonal atoms,
selected endpoint blocks, and the remaining zero block, which includes left-invalid, right-valid words.
The reflection is the identity on that zero block.  Each endpoint block has one of the following forms.
If neither endpoint is diagonal, use its raw graph reflection.  If a complete domain, complete
image, or both are diagonal, phase the entire endpoint block by $-1$.  At the initial partial
block, transpose the $q=0$ endpoints and phase both $q=1$ endpoints; at the final partial block,
transpose the $q=1$ endpoints and phase both $q=0$ endpoints.  Residual invalid or forbidden
atoms are diagonal phases.

The remaining first-round gate/input partial block has a rank-six support.  Pack its endpoint
sectors into the coordinates $(e,p,q)$ of Eq.~\eqref{eq:gate-input-partial-block}.
With the same $Q_{\mathrm{good}}$ and $D=C_e(U)$, let
\[
P_{\ker}=D\bigl(\ket+\!\bra+_e\otimes Q_{\mathrm{good}}\bigr)D^\dagger,
\qquad \Pi_{\mathrm{part}}=I-P_{\ker}.
\]
Then its reflection factors exactly as
\[
I-2\Pi_{\mathrm{part}}
=D \widehat H_e\bigl(Z_eZ_q\,\mathrm{CZ}_{e,q}\bigr)\widehat H_eD^\dagger.
\]
Here $\mathrm{CZ}_{e,q}=\operatorname{diag}(1,1,1,-1)$ in the computational basis of the ordered
pair $(e,q)$.
A fixed basis permutation packs the physical endpoints into $(e,p,q)$; the other seven bond bits
enter only through an equality selector.  This also proves directly the gate-input matrix
in Eq.~\eqref{eq:gate-input-partial-block}, without assuming that $U$ preserves a selected channel.

Let $P_\alpha$ denote the range projector in final support block $\alpha$.  These are the
orthogonal projectors supplied by the partition, not the generally nonorthogonal raw projectors.
Consequently
\[
\Pi:=\sum_\alpha P_\alpha,
\qquad
\prod_\alpha(I-2P_\alpha)=I-2\Pi,
\]
which is the reflection of the bondwise-merged projector.  A controlled reflection uses the same factors with one common
outer control.  In conjugated factors, the packing and $D,D^\dagger$ operations remain
uncontrolled and cancel in the control-zero sector; only the middle phase or transposition is
controlled.  Thus the resulting matrix is literally $\operatorname{diag}(I,I-2\Pi)$, with no
unrecorded global phase.

All phases, transpositions, and gatewise controls have fixed phase-sensitive words over
$\mathcal G$; in particular $Z=T^4$ and the displayed controlled phases reduce to the standard
exact Clifford--$T$ primitives.  The largest derived operation is an outer-controlled equality
phase on the ten data bits; its normalized multi-control construction uses eight clean
conjunction bits.  Every selector is computed, used only as a control, and reversed
before its input changes.  The same compute--use--uncompute order is retained under the outer
control, so all eight work qubits are initialized and returned exactly to zero in both control
sectors.
\end{proof}

\section{Weighted spectral mapping and sharpness of the quadratic bound}
\label{app:weighted-spectral-mapping}

For positive weights $w_j$, applying Lemma~\ref{lem:spectral} to $h'_j=w_jh_j$ gives the same
kernel identity and spectral bound with
\[
 H_w=\sum_jw_jh_j,\qquad
 Q_w=\sum_j\sqrt{w_j}\sqrt{h_j}\otimes\Gamma_j,\qquad
 W_w=\sum_jw_j\lVert h_j\rVert
\]
in place of $H_0,Q,W$.

Exact sparse access to the square roots is an additional input-model condition. The projector
families used below satisfy it directly because $\sqrt{h_j}=h_j$; for arbitrary positive weights,
the square roots must belong to the chosen exact input model.

Write $\Lambda_j=\lVert h_j\rVert$. Functional calculus on
$0\le t\le\Lambda_j$ gives
$h_j\preceq\sqrt{\Lambda_j}\sqrt{h_j}$. For a unit vector $\psi$ on system and
ancilla, set
\[
g_j=\bra\psi\sqrt{h_j}\otimes\Id\ket\psi\ge0,
\qquad
\Gamma_c=\sum_jc_j\Gamma_j
\]
for $c\in\mathbb R^M$. The Clifford relations imply
\[
\{Q,\Id\otimes\Gamma_c\}=2\sum_jc_j\sqrt{h_j}\otimes\Id.
\]
When $\lVert c\rVert_2=1$, $\Gamma_c$ is a Hermitian unitary. Cauchy--Schwarz,
followed by the supremum over such $c$, therefore yields
$\lVert g\rVert_2\le\lVert Q\psi\rVert$. Hence
\[
\bra\psi H_0\otimes\Id\ket\psi
\le \sum_j\sqrt{\Lambda_j}g_j
\le \sqrt W\,\lVert g\rVert_2
\le \sqrt W\,\lVert Q\psi\rVert.
\]
Applying this inequality to a ground state of $Q^2$ proves
$E_0(Q^2)\ge E_0(H_0)^2/W$ when $W>0$.

If $H_0\phi=0$, positivity forces $\sqrt{h_j}\phi=0$ for every $j$, and thus
$Q(\phi\otimes\chi)=0$ for every ancilla state $\chi$. Conversely, $Q\psi=0$
forces every $g_j$ to vanish. Positivity then gives
$(\sqrt{h_j}\otimes\Id)\psi=0$ for all $j$, proving
$\ker Q=\ker H_0\otimes\Fcal_{anc}$. When $W=0$, every $h_j$ vanishes and the
separate zero case follows directly.

For the first assertion of Proposition~\ref{prop:quad-tight}, take rank-one projectors
$h_1=\ket a\bra a$ and $h_2=\ket b\bra b$ on $\mathbb C^2$, with
$\braket ab=\cos t$ and $t\in(0,\pi/2]$. Since $\sqrt{h_i}=h_i$, a chirality basis
for the two ancillary Majoranas gives
\[
E_0(Q^2)=\min_{s=\pm1}\lambda_{\min}\!\bigl(H_0+s\,i[h_1,h_2]\bigr).
\]
Taking $\ket a=(1,0)$ and $\ket b=(\cos t,\sin t)$, both blocks have trace $2$
and determinant $\sin^4t$. Consequently,
\[
E_0(H_0)=1-\cos t=:\epsilon,
\qquad
E_0(Q^2)=1-\sqrt{1-\sin^4t}=2\epsilon^2+O(\epsilon^3).
\]
Thus $E_0(Q^2)/E_0(H_0)\to0$. This family rules out a universal linear bound;
it does not saturate the quadratic coefficient. For that second assertion, take the
one-term instance $h_1=\lambda\Id$ with $\lambda>0$. Then $W=\lambda$ and
$Q^2=h_1$, so $E_0(Q^2)=E_0(H_0)^2/W=\lambda$ exactly.

\section{Two-mode nilpotent extension}
\label{app:two-mode-nilpotentization}

Use the ordered occupation basis
$(\ket{00},\ket{01},\ket{10},\ket{11})$, where
$\ket{n_an_b}=(a^\dagger)^{n_a}(b^\dagger)^{n_b}\ket{00}$. Direct CAR evaluation gives
\[
s\ket{11}=\ket{00},
\qquad
s\ket{01}=\ket{10},
\qquad
s\ket{00}=s\ket{10}=0.
\]
It follows entry by entry that
$s^2=0$,
$ss^\dagger=\operatorname{diag}(1,0,1,0)$, and
$s^\dagger s=\operatorname{diag}(0,1,0,1)$. Thus the last two operators are
complementary projections, $\{s,s^\dagger\}=\Id_{ab}$, and $\lVert s\rVert=1$.
Both monomials in $s$ have even fermion parity.

The even operator $s$ commutes with every operator on the disjoint system. For Hermitian $Q$,
\[
(Qs)^2=Q^2s^2=0,
\qquad
\{Qs,(Qs)^\dagger\}
=Q^2(ss^\dagger+s^\dagger s)
=Q^2\otimes\Id_{ab}.
\]
This proves the full-space identities and the fourfold spectral multiplicity in
Theorem~\ref{thm:two-mode-nilpotentization}. If $k=\dim\ker Q$, then
$\operatorname{rank}(Qs)=2(d-k)$ because $\operatorname{rank}s=2$. Hence
\[
\dim\ker(Qs)=4d-2(d-k)=2d+2k,
\qquad
\dim H(Qs)=\dim\ker(Qs)-\operatorname{rank}(Qs)=4k.
\]
The coefficient, norm, and support statements follow term by term from
$\lVert q_cs\rVert=\lVert q_c\rVert\lVert s\rVert$ and from the two integer-coefficient
monomials of $s$. If $Q$ is odd, multiplication by the even $s$ preserves odd parity.
Sharing $a,b$ can destroy geometric locality of the terms, whereas the Hamiltonian identity
transfers any independently established locality of $Q^2$ without loss.


\section*{Acknowledgements}
\paragraph{Generative AI disclosure.}
OpenAI Codex (GPT-6) assisted with checking and correcting the manuscript, including
simplifying the proofs in Appendix~\ref{app:raw-source-gap}.
It was also used to draw schematic figures and prepare verification code.
The author takes responsibility for its content.

\printbibliography

\end{document}